%% file: arxiv.tex
\documentclass[11pt]{article}
\usepackage{lmodern}
\usepackage{dtrt}
\usepackage{dsfont}
\usepackage{algpseudocode}
\usepackage{environ}
\usepackage{comment}
\usepackage{siunitx}
\usepackage{titling}
\newcommand{\downarrowr}{\textcolor{red!70!black}{$\downarrow$}}
\newcommand{\uparrowg}{\textcolor{green!60!black}{$\uparrow$}}
\newcommand{\ambarrow}{\textcolor{orange!80!black}{$\updownarrow$}}
\input{definitions.sty}
\begin{document}
\setlength{\droptitle}{-2.5cm}
\title{The Free Option Problem of ePBS}

\author{Bruno Mazorra$^*$,
Burak Öz$^*$,
Christoph Schlegel$^*$, and
Fei Wu$^\dagger$
 \medskip
  \\
  \small
  $^*$Flashbots \\
  \small 
  $\dagger$King's College London \\
}

\date{Initial Version: September 20, 2025 \\
Current Version: January 8, 2026}

\maketitle

\thispagestyle{empty}

\begin{abstract}
Ethereum’s upcoming Glamsterdam upgrade introduces EIP-7732 \emph{enshrined Proposer–Builder Separation} (ePBS), which improves the block production pipeline by addressing trust and scalability challenges. Yet it also creates a new liveness risk:
builders gain a short-dated ``free'' option to prevent the execution payload they committed to from becoming canonical, without incurring an additional penalty. Exercising this option renders an empty block for the slot in question, thereby degrading network liveness.

We present the first systematic study of the free option problem. Our theoretical results predict that option value and exercise probability grow with market volatility, the length of the option window, and the share of block value derived from external signals such as external market prices. The availability of a free option will lead to mis-pricing and LP losses. The problem would be exacerbated if Ethereum further scales and attracts more liquidity. Empirical estimates of values and exercise probabilities on historical blocks largely confirm our theoretical predictions. While the option is rarely profitable to exercise on average (0.82\% of blocks assuming an $8$-second option time window), it becomes significant in volatile periods, reaching up to 6\% of blocks on high-volatility days—precisely when users most require timely execution. Moreover, builders whose block value relies heavily on CEX–DEX arbitrage are more likely to exercise the option. We demonstrate that mitigation strategies—shortening the option window or penalizing exercised options—effectively reduce liveness risk.

\end{abstract}

\hspace{0.5cm}

\input{arxiv_sections/intro}
\input{arxiv_sections/background}
\input{arxiv_sections/model}
\input{arxiv_sections/methdology}
\input{arxiv_sections/empiricals}

\input{arxiv_sections/mitigations}

\input{arxiv_sections/Conclusion}

\section*{Acknowledgments}
The authors appreciate Quintus Kilbourn, DataAlways, Potuz, Terence Tsao and Thomas Thiery for helpful discussions and comments. We also thank Danning Sui for providing data. Fei Wu acknowledges support from a Flashbots research grant.

\bibliography{bibliography} 
\appendix
\crefalias{section}{appendix}
\crefalias{subsection}{appendix}
\input{arxiv_sections/appendix}

\end{document}

%% file: arxiv_sections/intro.tex
\section{Introduction}
Ethereum's block production pipeline has evolved with Proposer–Builder Separation (PBS), 
currently implemented out-of-protocol via \emph{MEV-Boost} \cite{mevboost}, emerging as the leading paradigm for block building on Ethereum. In MEV-Boost, \emph{relays}, operated by centralized third parties, facilitate an auction by connecting proposers and builders. Relays forward builder bids and block headers for proposers to choose from and sign, before publishing the selected block at the slot start. While successful—with over 90\% of blocks outsourced to builders through MEV-Boost auctions \cite{mevboostpics}—this design relies on off-chain trust in relays, which can harm validators and threaten network liveness \cite{bloxroutebug}. 

To remove this dependency, EIP-7732, enshrined PBS (ePBS) \cite{eip7732}, is scheduled for inclusion in Ethereum’s upcoming Glamsterdam upgrade. Under ePBS, builders are staked on the network, and their bid is directly debited from their balance if they win the auction. The proposer publishes a beacon block and commits to the builder's block header (fixing the execution payload and blobs) at slot start, while the builder is granted additional time to reveal the full payload and blobs later in the slot. This removes the requirement for proposers to trust relays.

As another feature, ePBS introduces ``pipelining'' to scale Ethereum throughput. Builders are expected to deliver two types of data: a small execution payload requiring (EVM) computation, and larger blobs that do not necessitate much computation, but take longer to propagate. The scaling benefit is realized by separating the deadline by which the builder must deliver the execution payload from the deadline for blob delivery, which is shifted later into the slot. This allows more time for blob propagation without compromising the time allotted to payload execution.

However, this combination of features—making builders responsible for block distribution through enshrining the builder role in protocol,\footnote{Under MEV-boost, relays distribute blocks. To create a similar free option under MEV-boost, the builder and relays would need to collude. In ePBS, the builder can unilaterally decide to withhold block content.} and introducing dual deadlines—creates a new liveness threat: Builders obtain a short ``option window'' in which they can choose to prevent the block from becoming canonical through withholding blobs.\footnote{The builder can also make the block non-canonical by withholding the payload. However, as the blob deadline is substantially later, this is, as we argue, a less valuable option.} To build intuition why this might be lucrative, consider a casino where you can revert any roulette spin if the outcome is unfavorable to you. Such an option is extremely valuable as it offers a strictly positive risk-free expected return. ePBS grants the builder a similar option: if adverse information arrives after slot start—for example, unfavorable ETH price movements that makes the builder re-evaluate their on-chain trades—they can withhold a blob to make the payload non-canonical, thereby avoiding losses. The slot then remains empty, and no state transition occurs, degrading network liveness, on-chain market efficiency, and overall user experience, yet the builder faces no protocol-level penalty. This is known as the \emph{free option problem} of ePBS \cite{freeoption}.

In this paper, we provide the first systematic study of this free option problem. We make theoretical predictions about position sizes, on-chain prices, option values and exercise probabilities (cf. \Cref{sec:model_theory}). We construct our empirical methodology (cf. \Cref{sec:methodology}) and validate our predictions by examining them on historical Ethereum blocks (cf. \Cref{sec:empirical}). Finally, we propose possible mitigations to the problem and evaluate them empirically (cf. \Cref{sec:mitigations}). We summarize our contributions as follows:
\begin{enumerate}[topsep=2pt, itemsep=1pt, parsep=1pt, leftmargin=*]
   \item Option value and exercise probability increase with volatility, the length of the option window, and the value share of MEV generated from external market signals. The availability of a free option will lead to mis-pricing and LP losses.
    \item Assuming an $8$-second option time window, exercise probabilities are low on average in historical blocks (0.82\% missed block rate, 34.6\% of all missed blocks), but high during periods of high volatility (daily missed blocks of 5.5\% and 80\% of all missed blocks (cf.~\Cref{fig:liveness})).
    \item The option is more profitable in blocks where the value of CEX-DEX arbitrages contributes to a larger share of total block value. 
    \item We propose two mitigation strategies—shortening the option window and penalizing builders for exercising the option—and validate that both effectively reduce exercise probabilities.
\end{enumerate}

\begin{figure}[t]
  \centering
  \begin{subfigure}[t]{0.49\textwidth}
  \centering
  \includegraphics[width=\linewidth]{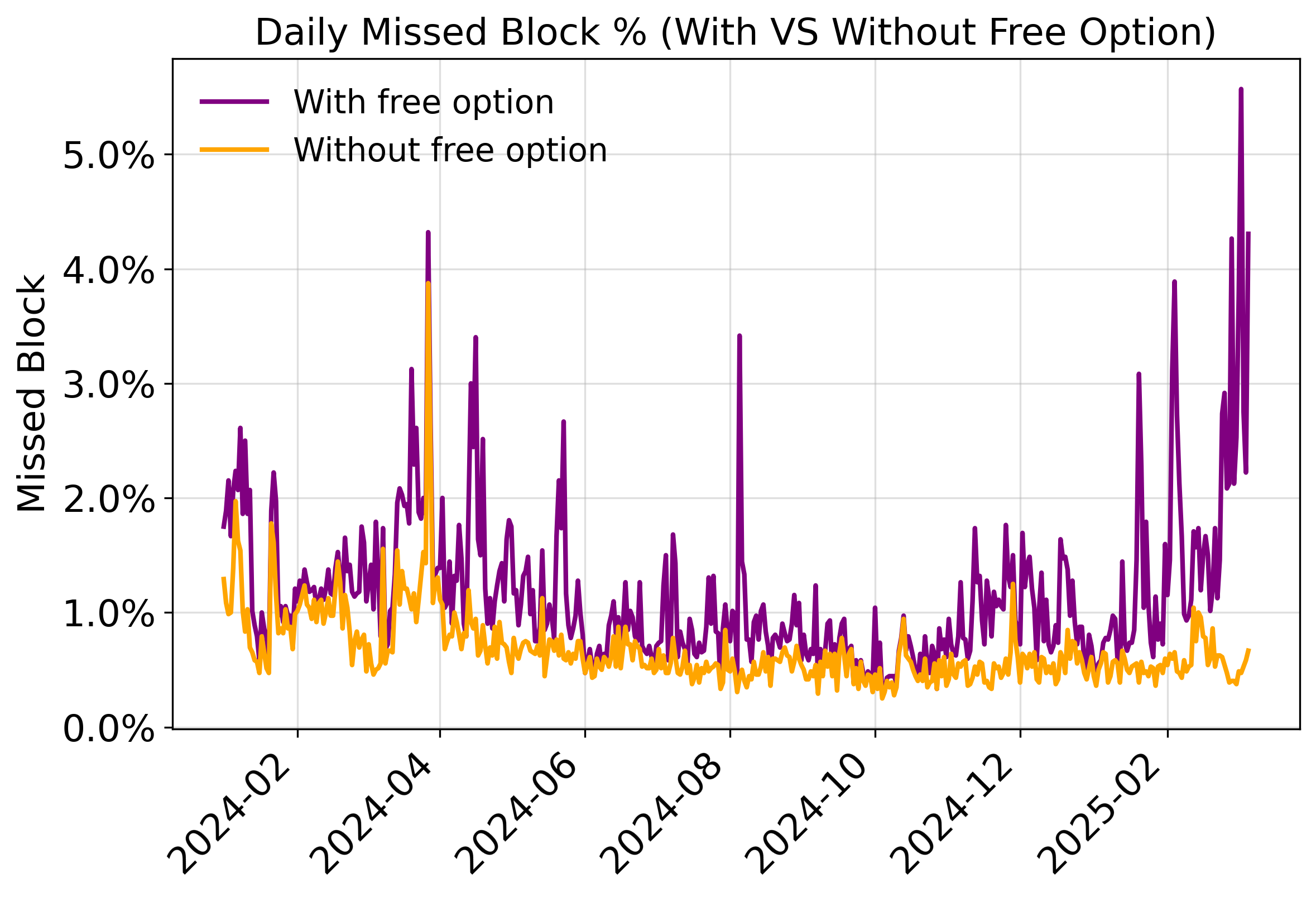}
  \end{subfigure}
  \begin{subfigure}[t]{0.49\textwidth}
  \centering
  \includegraphics[width=\linewidth]{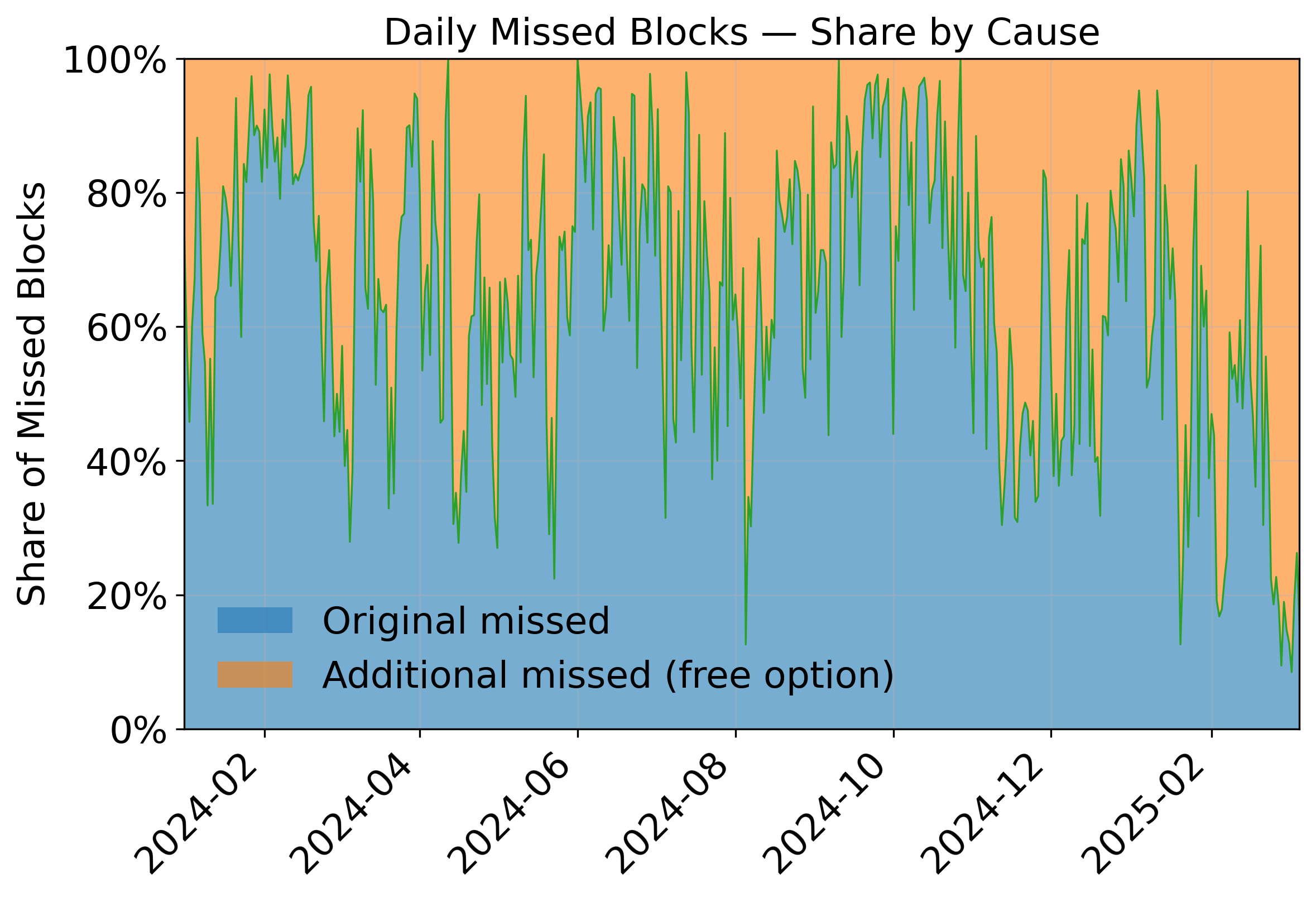}
  \end{subfigure}
  \caption{Daily missed block percentage with and without the ePBS free option and share of daily missed blocks by original causes and by the ePBS free option.}
  \label{fig:liveness}
\end{figure}

%% file: arxiv_sections/background.tex
\subsection{Background and Related Work}

\parhead{Ethereum.} Ethereum progresses in 12-second slots: in each slot, one validator (the \emph{proposer}) publishes a block, while others (\emph{attestors}) attest to the proposed block. A block consists of a consensus-layer \emph{beacon block} and an execution-layer \emph{payload}. The beacon block advances consensus and includes attestations and a cryptographic commitment to the execution payload, while the execution payload specifies the ordered transactions and their resulting state transitions. To improve Ethereum scalability, EIP-4844 upgrade \cite{4844} in March 2024 introduces \emph{blob (type-3) transactions}, which allow Layer-2 (L2) rollups to publish blobs of data on Ethereum Layer-1 (L1) for transaction settlement. Blob data is propagated through the consensus layer network and retained for a short period ($\approx$ 18 days); blob transactions propagated through the execution layer network include only references to associated blobs.

\parhead{Proposal-Builder Separation.} PBS decouples block production by allowing proposers to outsource block building to \emph{builders}. Today, PBS operates out-of-protocol via MEV-Boost \cite{mevboost}: builders submit blocks with bids to trusted intermediaries known as \emph{relays}, which forward only the block header with the highest bid to the proposer; once the proposer signs at the slot start, the relay reveals the full block. While MEV-Boost improves validator decentralization \cite{bahrani2024centralization}, it introduces off-chain trust and liveness dependencies in relays \cite{bloxroutebug}.

\emph{Enshrined Proposer-Builder Separation} (ePBS) or EIP-7732 \cite{eip7732} integrates the builder–\\proposer exchange into the protocol and eliminates relay trust. Specifically, with ePBS, the proposer publishes a beacon block at the slot start ($t=0$s) with the commitment to the header of the execution payload built by the winning builder. Builders are staked on the mainnet, and the winning builder's bid is directly removed from their beacon chain balance. After the beacon block is attested by $t=1.5$s, the winning builder subsequently reveals the full payload. A \emph{Payload Timeliness Committee} (PTC)—a subset of attestors—votes on whether the payload and associated blob data arrived in time \cite{ptc,ElowssonDualPTC}. At $t = 4$s, the PTC first enforces the \emph{payload deadline}: if the execution payload is not observed, the slot is flagged as empty and no state transition occurs. At $t = 10$s, the PTC enforces the \emph{blob deadline} and votes positively only if both deadlines are satisfied; otherwise, missing blobs or insufficient votes again prevent the payload from becoming canonical and leave the slot empty. \Cref{fig:slot_pipeline} illustrates this pipeline.\footnote{The time indications follow the dual-deadline design in~\cite{ElowssonDualPTC}, which remains subject to change.} The design also improves the network scalability by allowing more time for payload execution and blob propagation. 


\smallskip
\parhead{Maximal Extractable Value.} 
Maximal Extractable Value (MEV) refers to the additional revenue that block participants can capture by strategically ordering, including, or excluding transactions \cite{daian2020flash}. MEV can be extracted purely from on-chain state (atomic MEV) \cite{qin2021liqui,qin2022quantifying}, or by combining on-chain and cross-domain signals (non-atomic MEV) \cite{obadia2021unity,liobanonatomic,burakcrosschainarb}, such as prices on a centralized exchange (CEX). Under PBS, MEV is primarily realized by specialized \emph{searchers}, who bribe builders through priority fees and transfers alongside their transactions \cite{mamageishvili2024searcher}. A prominent category is arbitrages between a centralized and a decentralized exchange (DEX), which, despite occupying only about 2\% of block space, contributes roughly 20\% of block value \cite{whowinsandwhy,liobanonatomic}. 
In practice, most leading CEX–DEX searchers are integrated with a major block builder \cite{liobanonatomic,whowinsandwhy,wu2025measuringcexdex}, consolidating both execution and block building advantages.

\smallskip
\parhead{Related Work.}
Previous research on PBS spans empirics and game-theoretic modeling. \cite{bahrani2024centralization} demonstrated that PBS enhances validator decentralization but introduces centralization concerns among builders. Empirically, \cite{schwarz2023time,heimbach2023ethereum} analyzed the post-PBS block-building market, highlighting emerging concentration and censorship pressures. \cite{gupta2023centralizing,liobanonatomic} demonstrated the advantage of builders integrated with CEX-DEX searchers in block building at times of high volatility. Together with \cite{liobanonatomic,whowinsandwhy}, \cite{wu2025measuringcexdex} provided empirical evidence that leading block builders are integrated with a CEX-DEX searcher. The work further revealed that profit-sharing arrangements between searchers and builders are closely tied to their degree of integration. The downstream effects of vertical integration and builder concentration were further examined in \cite{yang2025decentralization}, which studied implications for proposer revenue and block optimality. Finally, \cite{schwarz2023time,oztime2023} revealed the proposer's strategic behavior of delaying block proposal to optimize reward under PBS, which can potentially contribute to missed slots and affect network liveness.

%% file: arxiv_sections/model.tex
\section{The Free Option Problem}
\label{sec:model_theory}

Under the current MEV-Boost architecture, the relay releases the full block at the slot start, leaving the builder no discretion after the proposer's commitment to the block. In ePBS, by contrast, the winning builder can unilaterally decide not to reveal the execution payload or the blobs by their respective PTC deadlines, thus rendering the slot empty.

When the block value is solely determined by the Ethereum L1 state, 
the builder has no incentive to withhold under the unconditional bid. However, if the builder's reward depends on cross-domain information that may arrive in the window between the proposer commitment and the PTC deadlines, the builder can condition the reveal decision on the late information. 
To do this, the builder prepares a blob but keeps it private, retaining the option to release or withhold it by the deadline. Since the bid is due regardless, the builder reveals when inclusion is favorable and withholds otherwise, treating the bid as a sunk cost. Thus, ePBS effectively grants the builder a short-dated \emph{free option} to prevent the payload from becoming canonical by withholding the payload or blobs after commitment, \emph{with no additional protocol penalty}. \Cref{fig:slot_pipeline} illustrates the window granted to the builder to decide whether to exercise the option. Given that currently blobs need roughly 2 seconds to be propagated \cite{migalabs}, the builder has a window of roughly 8 seconds to decide whether to exercise the free option.

\begin{figure}[t]
    \centering
    \includegraphics[width=1\linewidth]{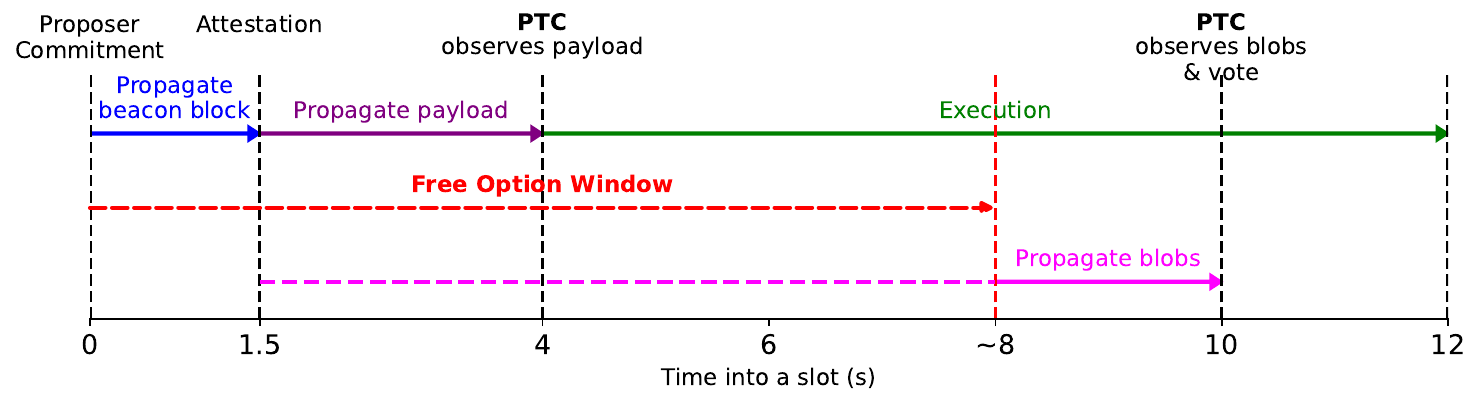}
    \caption{Slot pipeline and PTC deadlines under ePBS. Red markers indicate the free option window, which extends until the latest point when a strategic builder can release blobs while still ensuring the payload remains valid.}
    \label{fig:slot_pipeline}
\end{figure}

\paragraph{The free option with PeerDAS.}\label{appendix:PeerDas}
The free option strategy could take a different shape after Fusaka hard fork due to PeerDAS \cite{wagner2024documentation}. Instead of being forced to propagate $100\%$ of the blob data to guarantee data availability, the builder can strategically release only a subset of columns of the extended blob upfront (e.g., just below the threshold required for reconstruction) and withhold the rest. The builder then decides whether to release the other columns close to the PTC deadline. If they choose to release, it may suffice to propagate only a small additional fraction of the columns of the blob, which already enables full recovery even in the most adversarial setting. Otherwise, the builder simply withholds. In short, PeerDAS reduces the amount of data that must be delivered at the last moment, potentially shrinking the 2 seconds needed for propagating the blob to $\lesssim 400$ms plus reconstruction time,\footnote{Depending on reconstruction duration, the builder may alternatively release all remaining columns to minimize reconstruction time, at the cost of higher propagation delay.} and so, effectively increasing the length of the free option window.

\subsection{Trading Volatile Assets with an External Market}
Arguably, the most valuable use of the free option would be to revert DEX trades in case the position value moves against the builder. More specifically, a builder can add DEX transactions to the block and watch the price of the traded assets on an external market (usually a CEX). If it would no longer be worthwhile to execute the DEX trade(s), according to the external price signal, at the end of the option window, the builder can revert them by withholding the blobs to prevent the block from becoming canonical.

Placing the trade in the block with the option to revert is effectively equivalent to holding a short-dated option on the traded pair (instead of holding the spot/resp. a short-dated future in case of a block without optionality).  However, in comparison to a regular option traded in financial markets, the option is not restricted to a fixed notional value. Instead, the builder can, in principle, trade arbitrarily large positions on the DEX to construct an option with arbitrarily large notional. If liquidity close to the current market price would be unbounded this would give him arbitrarily large option value. In reality, however, on-chain liquidity is bounded and the value of the option is effectively bounded by the {\bf liquidity} available on chain. As option value is increasing in {\bf volatility}, the more liquidity there is for volatile assets, the more valuable the option is and the more likely it is to be exercised. Given these considerations, we subsequently derive the conditions for optimally sizing the DEX trade(s) and the resulting DEX prices, option value and exercise probability and establish how volatility, liquidity and the amount of atomic MEV influence them.

We assume for simplicity that there is only one risky asset that can be traded for the numéraire on CEXes and DEXes.  The arguments would generalize to multiple assets.
We can decompose the profit at time $\tau$ of proposing a block at time $0$ as:
$$\Pi_\tau(y)=\mu+(1+r_\tau) y-P_{DEX}(y/P_0)=\Pi_0(y)+r_{\tau}y$$
where $\mu$ is the block value net of value generated from the position in the risky asset, $y$ is the position size, measured in the numéreire at CEX prices, of the risky asset traded on the DEX, $r_t:=\frac{P_t-P_0}{P_0}$ is the return of the risky asset (on the CEX) from time $0$ until $t$ and $P_{DEX}(x)$ is the total cost of the DEX trade at time $0$ as a function of trade size $x=y/P_0$. 

We assume that $E[r_t]=0$ but otherwise put no restrictions on the return distribution. Moreover, we assume (w.l.o.g.) that $P'_{DEX}(0)\leq P_0$, i.e. the risky asset is correctly priced or under-priced on the DEX. Analogous considerations would work in the case of overpricing, $P'_{DEX}(0)\geq P_0$, in which case the builder would sell rather than buy the risky asset on the DEX.

The builder solves:
$$V^*:=\max_{y}E[\max\{0,\Pi_{\tau}(y)\}]$$

which gives a corresponding {\bf exercise probability}
$$P^*=Pr[\Pi_{\tau}(y)<0].$$
We call the difference between the value and the block value without option (i.e. $V^*-\max_y \Pi_0(y)$), the {\bf net option value}.
\begin{remark}
The counter-party to a DEX trade would usually be the DEX liquidity providers (LP). In this case, the value captured through the option would be at the expense of both LPs and other arbitrageurs: Part of the arbitrage value due to stale DEX prices that is otherwise captured by arbitrageurs in the next block is now captured by the builder exercising the option. LPs forego fee income if price discrepancies are (partially) arbitraged through option exercise. 
\end{remark}
First we observe that at time of block commitment, with the option, DEXes will over-price the risky asset post-trade. 
\begin{proposition}\label{overprice}
The DEX-price post trade (at $\tau=0$) is strictly greater than the CEX  market price.
\end{proposition}
\begin{proof}
The first order condition for the builder is
$$\tfrac{P_{DEX}'(y^*/P_0)}{P_0}=1+E[r_\tau|\Pi_{\tau}(y^*)>0]>1.$$
\end{proof}

\begin{example}
If returns are normal with mean $0$ and volatility of $\sigma$, then
$$E[r_\tau|\Pi_\tau(y)>0]=\sigma\frac{\phi(z^*)}{1-\Phi(z^*)}\equiv\sigma\lambda(z^*)$$
where 
$z^*:=\frac{\Pi_0(y^*)}{\sigma y^*}$ and $\lambda(\cdot)$ is the hazard rate of the standard normal. If the AMM is a CPMM, the DEX price can be approximated~\cite{schlegel2025arbitrageboundedliquidity} by $$P_{DEX}(x)\approx  P'_{DEX}(0)x+P'_{DEX}(0)\tfrac{x^2}{L/P_0}\Rightarrow P'_{DEX}(x)\approx P'_{DEX}(0)+P'_{DEX}(0)\tfrac{2x}{L/P_0}$$ where $L$ is the value of the risky asset reserves of the AMM measured in the numéraire at CEX prices. If $\mu=0$ and $P'_{DEX}(0)=P_0$ the optimal position is
$$\lambda(z^*)=2z^*\Rightarrow y^*=z^*\sigma L\approx0.61\sigma L.$$
Thus, the position is linearly increasing in volatility and in AMM liquidity, and the price difference between the AMM price post-trade and the CEX price is
$$\tfrac{P'(x^*)-P_0}{P_0}=1.22\sigma.$$
More, generally, define the pre-trade price gap between the CEX and the DEX by $\delta:=\frac{P_0-P_{DEX}'(0)}{P_0}$. We can write the first order condition as
$$\lambda(z^*)+\tfrac{\delta}{\sigma}=2z^*.$$
Let $z_0:= 0.61$ be the solution for $\delta=0$. We can approximate:
$$z^*\approx z_0+\tfrac{1}{2-\lambda'(z_0)}\tfrac{\delta}{\sigma}\Rightarrow y^*\approx 0.61\sigma L+0.8\delta L.$$
Thus, the position is approximately linearly increasing in AMM liquidity and scales approximately linear in a weighted ratio of the volatility and the pre-trade price gap. 
\end{example}
 In general, we have the following findings which are straightforward to prove (see \Cref{appendix:proofs} for details):

\begin{proposition}\label{prop1:liquidity}
The value $V^*$ of the option and the probability of exercise $P^*$ increase in liquidity. 
\end{proposition}

Next we have that more dispersed returns (e.g. through higher volatility) lead to higher option value: 
\begin{proposition}
\label{prop2:return}
If returns are more dispersed in the sense that they are a mean-preserving spread, the value $V^*$ of the option increases. As dispersion is increasing in time, its value is increasing in the time to expiration.
\end{proposition}

Monotonocity of exercise probability with dispersion does not hold in this full generality, but is generically true, as for example for normal returns:
\begin{proposition}\label{prop3:vol}
If returns are normally distributed with mean $0$ and volatility $\sigma$ and the DEX follows a constant product market maker, then the exercise probability $P^*$ is increasing in $\sigma$.
\end{proposition}
\begin{remark}
How the return distribution scales with slot time is a contested question. In theory, volatility scales with the square root of time (assuming a Gaussian process). Practically, in short time intervals, it scales between the square root and linearly in time.
\end{remark}
Finally, the option value increases and the exercise probability decreases in the amount of atomic MEV on chain:
 \begin{proposition}
 \label{prop4:atomic_mev}
 The option value $V^*$ is increasing and the exercise probability $P^*$ is decreasing in $\mu$, the atomic MEV on-chain. 
 \end{proposition} 



 The free option problem gets worse if we reduce 
 atomic MEV per unit of liquidity on-chain. As reducing on-chain MEV is a goal that many solutions and applications work towards (e.g. through better AMM design, middle-layer infrastructure, encrypted mempools), paradoxically, the free option problem gets exacerbated the better we make the on-chain trading experience.

\subsection{Prediction Markets}
It should be noted that, while DEX trading with an external price signal would be the most important current use case of the free option, the availability of a free option could interfere with other applications as well. Moreover, it could prevent some applications from gaining much traction on Ethereum, because the more users it attracts, the more vulnerable it would become to exploitation through the free option.
As an example, imagine a prediction market on Ethereum 
would attract liquidity comparable to the most popular prediction markets elsewhere. Suppose a highly traded “Which party wins the 2028 US Presidential Election?” contract is about to be influenced by a pivotal news expected sometime in the next few minutes (or hours). A searcher could buy all the liquidity of one side of the contract in a candidate block, then wait until just before the reveal deadline. If either the information is not released or the news are against the trade, exercise the option; otherwise, reveal the block. Because information arrival is stochastic, as long as there is enough liquidity with respect to the information arrival time window, this strategy can be repeated every slot until the news finally lands. With sufficient liquidity, the same pattern could arise around football matches, an eccentric billionaire's number of tweets, or any other event with an expected important announcement or signals.




%% file: arxiv_sections/methdology.tex
\section{Methodology}
\label{sec:methodology}

In this section, we outline the methodology for empirically examining the free option on historical blocks, including collecting data, measuring block value and the position value of DEX trades, and constructing the option metrics.

\parhead{Data collection.}
We curate a dataset covering all the blocks that contain at least one relevant DEX trade between January 1, 2024, and March 8, 2025. 
 Using the approach introduced by \cite{wu2025measuringcexdex}, we identify as relevant, \SI{6930821} CEX-DEX searcher transactions from 23 major searchers in \SI{1877729} blocks, including information on searcher payments and token amounts. We then collect the total value and the builder data of these blocks from \cite{relayscan}. To estimate position values of DEX trades, we collect Binance historical quotes data from Tardis.dev \cite{tardisdata}, tracking mid prices for USDT pairs of all traded tokens. Prices are sampled for 8 seconds after the slot time and converted into ETH using contemporaneous ETH/USDT mid price, leveraging the superior liquidity of USDT pairs to ensure consistent and accurate pricing.

\parhead{Total block value.}
We define the total value of a block as the sum of all payments to the builder at the time of commitment, including all included transactions' priority fee and coinbase transfer.
Formally, for a transaction $i \in T_b$, where $T_b$ represents the set of transactions included in block $b$, we denote its priority fee (tip) and coinbase transfer as $i_\text{tip}$ and $i_\text{transfer}$, respectively. The total value of block $b$, denoted by $v_b$, is

\[
v_b \;=\; \sum_{i \in T_b} (i_\text{tip} + i_\text{transfer}).
\]

\parhead{DEX position value.}
The value of a DEX trade is the revenue that the searcher can make by flattening their DEX-acquired inventory on the CEX, net of any trading fees. 
As the searcher's actual attainable execution price on CEX is not observable, we estimate their revenue by using \emph{markouts} calculated with CEX mid prices.
Formally, consider a DEX transaction $j$ that purchases $x$ amount of token A and sells $y$ amount of token B on the DEX, net of liquidity provider fees. Let $P_A(t)$ and $P_B(t)$ be the respective CEX prices converted to ETH at markout horizon $t$. The ETH value of trade $j$ at time $t$ can be estimated as:
\[
\pi_j(t) = x \cdot P_A(t) - y \cdot P_B(t) - \text{base fees}_j - \text{CEX taker fees.}\footnotemark\
\]\footnotetext{We assume CEX-DEX searchers operate at the highest user tier on Binance, thereby benefiting from the lowest applicable fee rate of 0.01725\% \cite{cexfees}.}

\parhead{Option metrics.} 
Following \cite{liobanonatomic,whowinsandwhy,wu2025measuringcexdex}, we consider the builder and the CEX-DEX searchers as an integrated entity. 
Assuming the payments from CEX-DEX searchers to the builder at commitment reflect the contemporaneous values, 
the  value of block $b$ at time $t \in [0, \tau]$ after commitment becomes
\[
  \Pi_b(t)= v_b - 
  \sum_{j \in J_b} \bigl(j_{\text{tip}} + j_{\text{transfer}}\bigr) +
  \sum_{j \in J_b} \pi_j(t),
\]
where $J_b \subset T_b$ indexes the DEX trades included in block $b$ and $\tau \approx 8$s. 
The \emph{net option value} realized by the builder (hereafter, \emph{option value}) is
\[
\max\{0, \Pi_b(t)\}-\Pi_b(t)=\max\{0,-\Pi_b(t)\} .
\]

\parhead{Case Study: Slot \texttt{10990298}}
To illustrate how the block value and option value evolve within a slot, we conduct a deep dive into the block proposed at slot~\texttt{10990298} (block~\texttt{21776075}) by builder \texttt{Titan}.

In this slot, the total block value is $v_b = 0.0659$~ETH, of which 0.0311~ETH comes from relevant DEX trades. As shown in \Cref{fig:example_block}, at the slot start, the block value is $\Pi_b(t)=0.0581$~ETH, meaning that Titan is initially in profit after accounting for relevant DEX trade positions.\footnote{In particular, this trade by CEX-DEX searcher \texttt{Wintermute}: \url{https://etherscan.io/tx/0x18a092ccc78603a98e8992140fd5e1ff79cbfc910c2730359ce819ea9cc74012}.} As time progresses and the ETH/USDC price declines, the block value deteriorates and becomes unprofitable around 3.5 seconds into the slot. By the end of the free option window, the loss reaches 0.064~ETH due to the adverse DEX trade positions. 
With ePBS, \texttt{Titan} could have exercised the free option and thereby avoided 0.064~ETH loss.


\begin{figure}[t]
    \centering
    \includegraphics[width=0.55\linewidth]{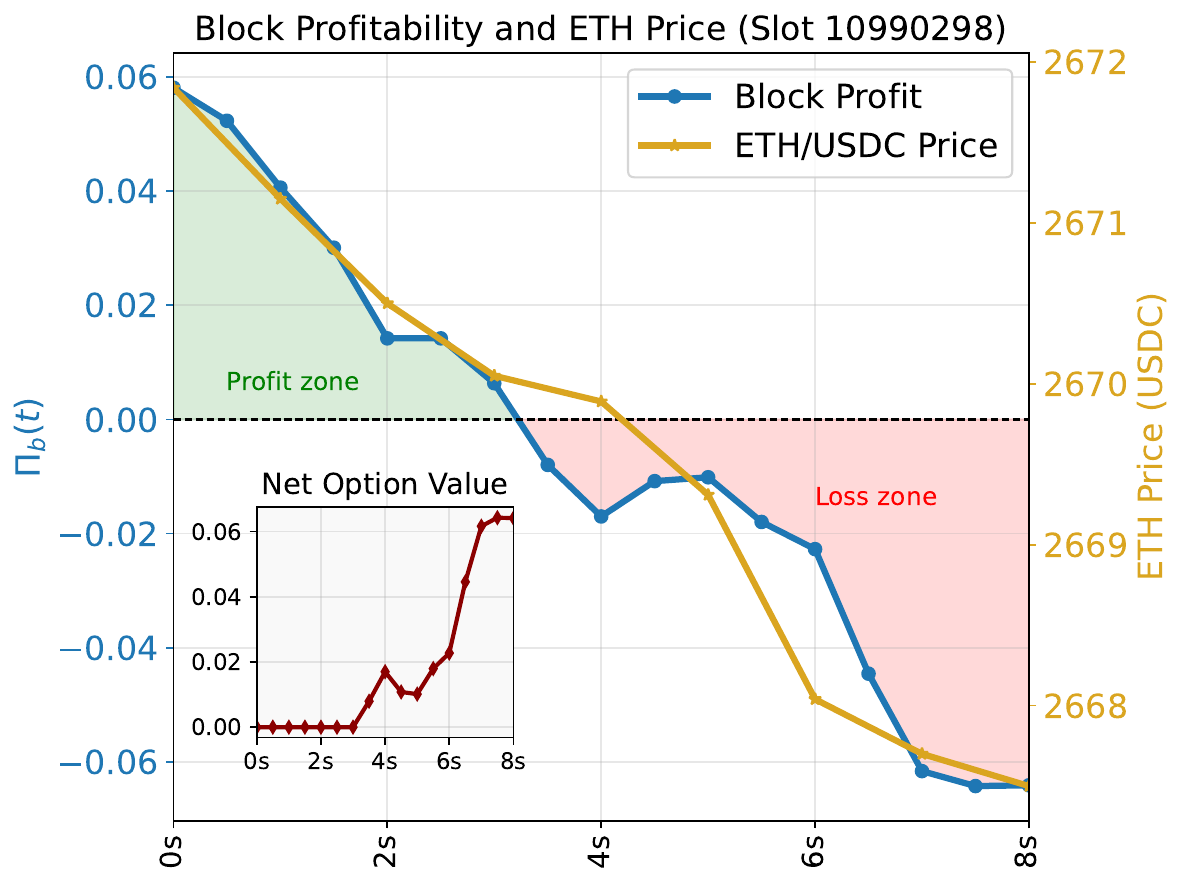}
    \caption{Evolution of block value $\Pi_b(t)$ and ETH/USDC price during slot~\texttt{10990298}. The inset depicts the option value.}
    \label{fig:example_block}
\end{figure}

\parhead{Limitations.}
Because we analyze historical blocks from a period without ePBS, integrated searcher-builders may have submitted different trades and built different blocks had the option been available.\footnote{For example, our theoretical results would predict that with an option they would make DEX-trades even if the DEX pairs are correctly priced and would generally trade larger positions such that the DEX price post trade overshoots the CEX price.} This counterfactual likely leads us to \textit{understate} option values and exercise probabilities. Additionally, our methodology has several limitations that may bias the measurement of option metrics:
\begin{enumerate}[topsep=1pt, itemsep=1pt, parsep=1pt, leftmargin=*]
\item We aggregate all searchers' trades in the block and proxy the transaction value at commitment with the searcher payment. While reasonable for integrated searcher-builders \cite{wu2025measuringcexdex}, this may understate true value. Our post-commitment estimates using Binance mid prices neglect liquidity costs and price impact. Both effects bias the option value downward.
\item Conversely, historical Binance prices may already reflect the impact of trades themselves. Large trades can move CEX prices substantially, meaning some extreme option values can be driven by self-imposed price impact. This makes our option value estimates potentially upward biased (cf. \Cref{appendix:high_volatility_days}).
\item Given available resources, we restrict analysis to trades by 23 major searchers involving tokens listed on Binance. Smaller searchers and tokens traded on other CEXes are excluded. Based on a Dune Analytics query adapted from \cite{cexdexdune}, this exclusion accounts for roughly 7\% of total volume. We do not account for other cross-domain signals (e.g., cross-chain arbitrages \cite{burakcrosschainarb}) that could also affect the option metrics.
\end{enumerate}

%% file: arxiv_sections/empiricals.tex
\section{Empirical Findings} \label{sec:empirical}
In this section, we present the empirical results of the option metrics examined on historical Ethereum blocks. Overall, our findings are broadly consistent with the theoretical predictions in \Cref{sec:model_theory}.

\subsection{Option Exercise Probability and Value}
We first examine the probability that the option is profitable to exercise. Overall, for an 8-second option window, the option is profitable in only 0.82\% of all observed blocks, but this average masks substantial variation across time. As shown in \Cref{fig:exercise_prob_hist}, exercise probabilities are low for most days, yet on several days with high price volatility, the probability exceeds 3.5\%, causing over a total of 4\% of blocks scheduled on those days to be missed, thus significantly degrading network liveness (cf. \Cref{fig:liveness}). This relationship is confirmed in \Cref{fig:exercise_prob_over_time}, where spikes in exercise probability align closely with periods of elevated ETH price volatility. Similar to \cite{liobanonatomic}, we calculate the ETH price volatility during a period by $\log_{10}(\frac{P_\text{high}}{P_\text{low}})$, where $P_\text{high}$ and $P_\text{low}$ are the highest and lowest prices during that period, respectively.


\begin{figure}[t]
\centering
\begin{subfigure}[t]{0.49\textwidth}
  \centering
  \includegraphics[width=\linewidth]{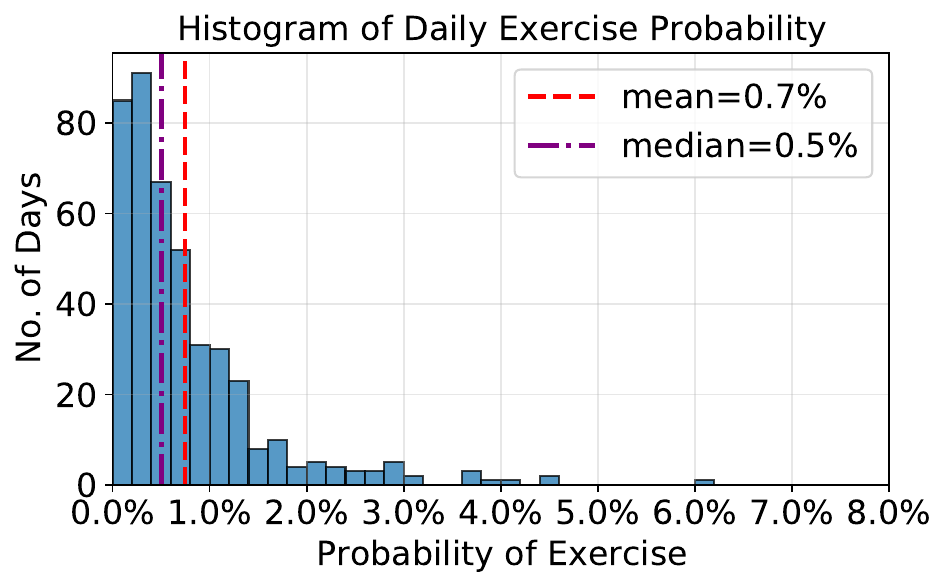}
  \caption{}
  \label{fig:exercise_prob_hist}
\end{subfigure}
\begin{subfigure}[t]{0.49\textwidth}
  \centering
  \includegraphics[width=\linewidth]{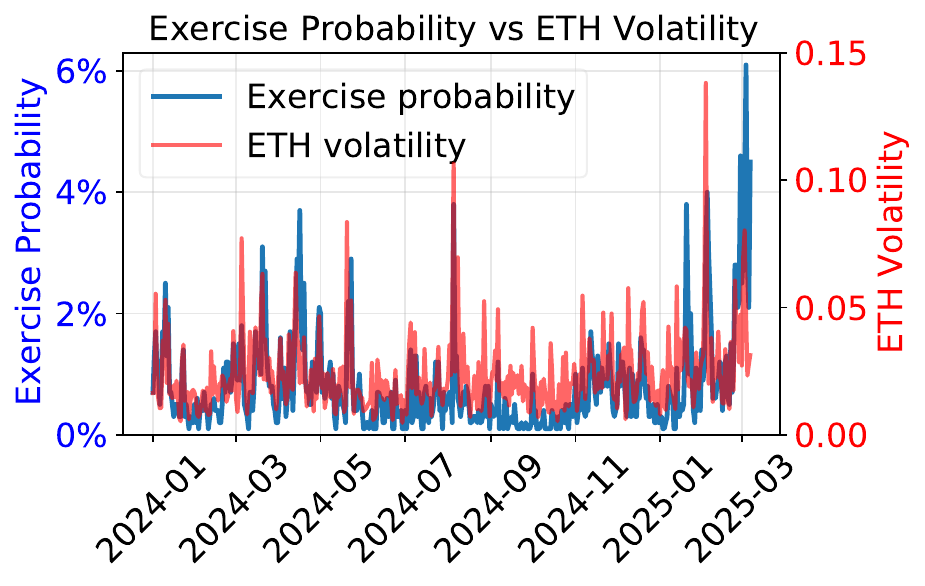}
  \caption{}
  \label{fig:exercise_prob_over_time}
\end{subfigure}
\begin{subfigure}[t]{0.49\textwidth}
  \centering
  \includegraphics[width=\linewidth]{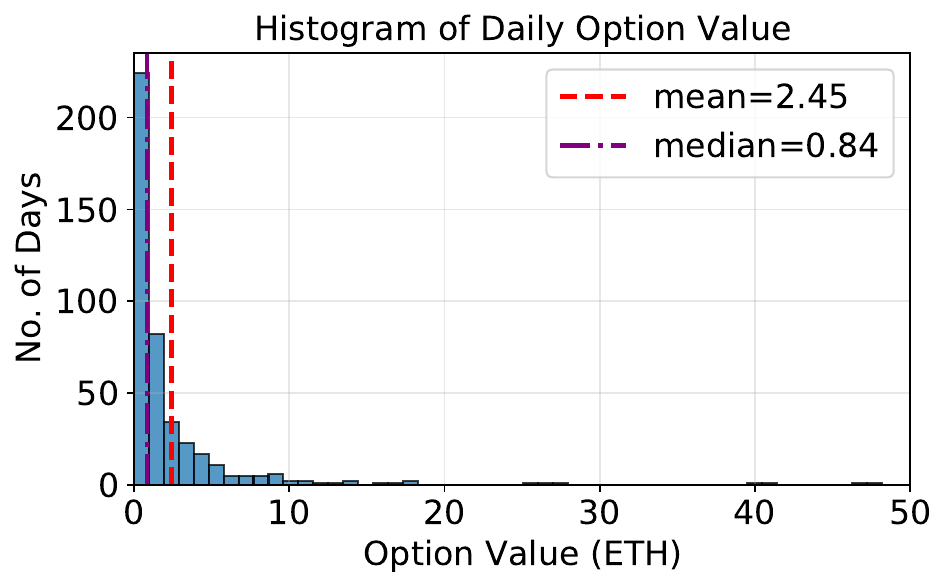}
  \caption{}
  \label{fig:option_value_hist}
\end{subfigure}
\begin{subfigure}[t]{0.49\textwidth}
  \centering
  \includegraphics[width=\linewidth]{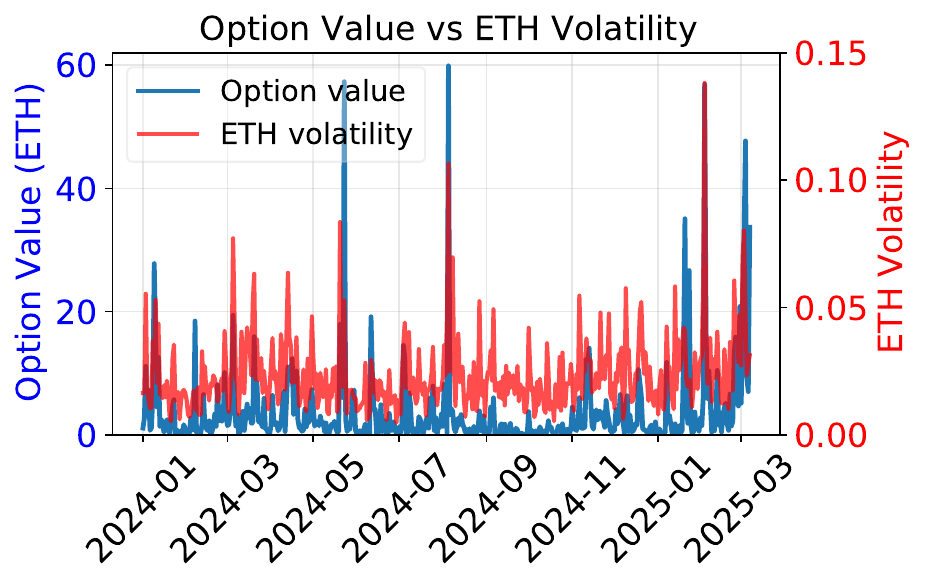}
  \caption{}
  \label{fig:option_value_over_time}
\end{subfigure}
\caption{(a) Histogram of daily option exercise probability. (b) Daily option exercise probability and ETH price volatility. (c) Histogram of daily aggregate option value. (d) Daily aggregate option value and ETH price volatility.}
\end{figure}

Turning to option value, i.e., the value captured from exercised options
, we again observe low averages but highly skewed distributions. The asymmetry is evident in \Cref{fig:option_value_hist}, where the median daily option value is 0.84 ETH and most days cluster near zero, but several exhibit extreme values. As is shown in \Cref{fig:option_value_over_time}, on high-volatility days, the option value is substantially higher, sometimes exceeding 40 ETH and accounting for the bulk of total option value.\footnote{The daily exercise probability and aggregate option value reflect only the blocks included in our sample. Consequently, some days contain more blocks than others.}


These results confirm \Cref{prop2:return,prop3:vol}: volatility amplifies the exercise probability and the option value. Moreover, we notice the option value comes from only a few short intervals on high-volatility days. Detailed case studies of such episodes are provided in \Cref{appendix:high_volatility_days}.

\subsection{Builder Heterogeneity}
We next examine cross-sectional heterogeneity across builders. Builders differ systematically in order flow composition: large builders access a diverse mix of order flow, while small builders often rely heavily on a single type of order flow from their exclusive providers \cite{whowinsandwhy}. This distinction is crucial, as theory predicts that the total block value's dependence on CEX-DEX flow increases the option value and exercise probability.

Specifically, for the two market leaders, \texttt{beaverbuild} and \texttt{Titan}, the option is profitable in just 0.75\% and 0.66\% of their blocks, respectively, consistent with their access to rich non-CEX–DEX order flow. By contrast, smaller builders whose block value is dominated by CEX–DEX flow exhibit much higher exercise probabilities: \texttt{blockbeelder}, \texttt{blocksmith}, and \texttt{gigabuilder} show rates of 9.39\%, 6.86\%, and 23.44\%, respectively. Note that because their market shares are small, their contribution to the overall incidence of exercised options is limited. In contrast, large builders, though less likely to exercise, account for most exercised options in aggregate simply by virtue of their market share.

The trajectory of \texttt{rsync} illustrates this dynamic further. Before September 2024, when they were a large builder with 23.43\% market share and CEX–DEX flow accounted for 29\% of their block value, their option exercise probability was 1.09\%. After retreating from the builder market and relying more heavily on CEX–DEX searcher activity \cite{wu2025measuringcexdex}, CEX-DEX flow accounted for 56\% of their block value, and their exercise probability rose to 3.22\%. We summarize the results for representative builders in~\Cref{tab:builders}.

\begin{table}[t]
\centering
\caption{Cross-sectional heterogeneity analyses across builders.}
\label{tab:builders}

\begin{tabular}{lccc}
\toprule
{Builder} & {Market Share} & {Mean CEX-DEX Flow} & {Option Exercise}  \\
& & {Value Share} & {Probability} \\
\midrule
Titan	& 29.04\% & 14.90\% &	0.66\% \\
beaverbuild	& 49.19\% & 19.12\% &	0.75\% \\
rsync (before 2024-09-01)	& 23.43\% & 29.01\% &	1.09\% \\
rsync (after 2024-09-01)	& 5.55\% & 56.25\% &	3.22\% \\
blockbeelder	& 0.09\% & 71.89\% & 9.39\% \\
blocksmith & $<0.01\%$ & 79.43\% & 6.86\% \\
gigabuilder & $<0.01\%$ & 84.57\% & 23.44\% \\
\bottomrule
\end{tabular}
\end{table}

\begin{figure}[t]
  \centering
\begin{subfigure}[t]{0.49\textwidth}
  \centering
  \includegraphics[width=\linewidth]{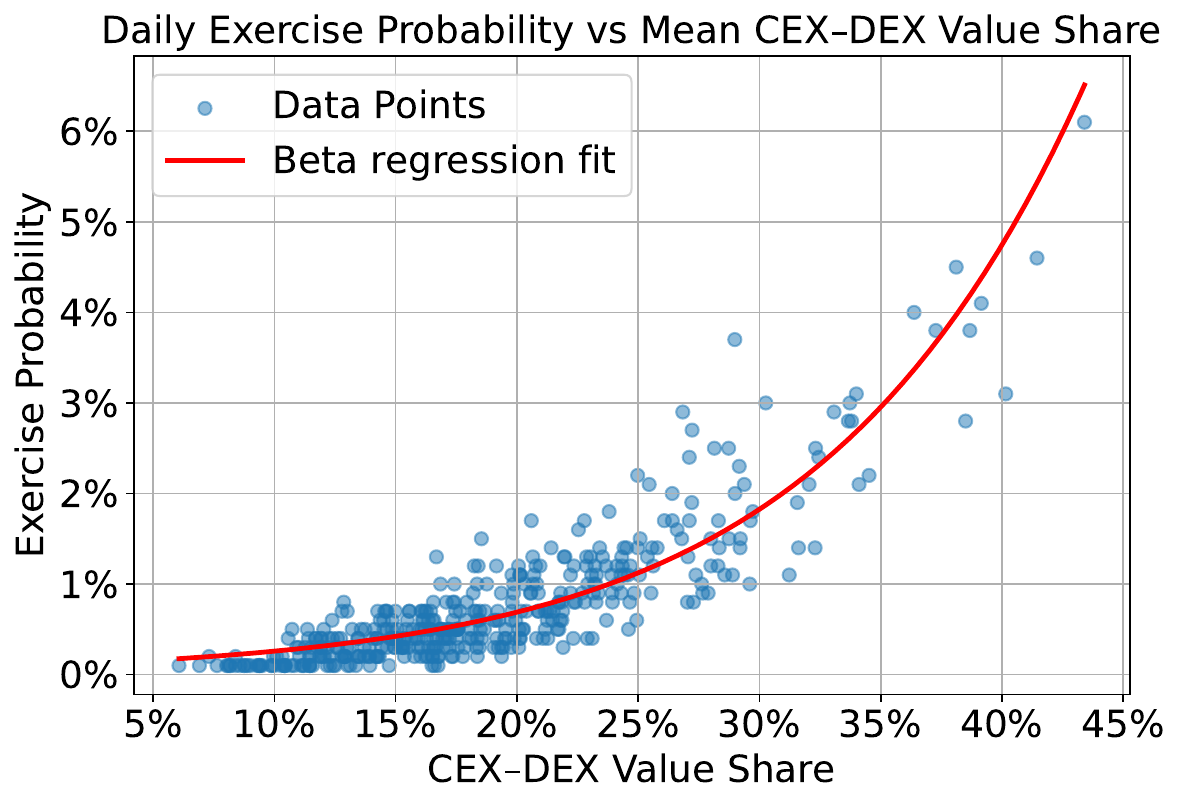}
  \caption{}
  \label{fig:beta}
\end{subfigure}
\begin{subfigure}[t]{0.49\textwidth}
  \centering
  \includegraphics[width=\linewidth]{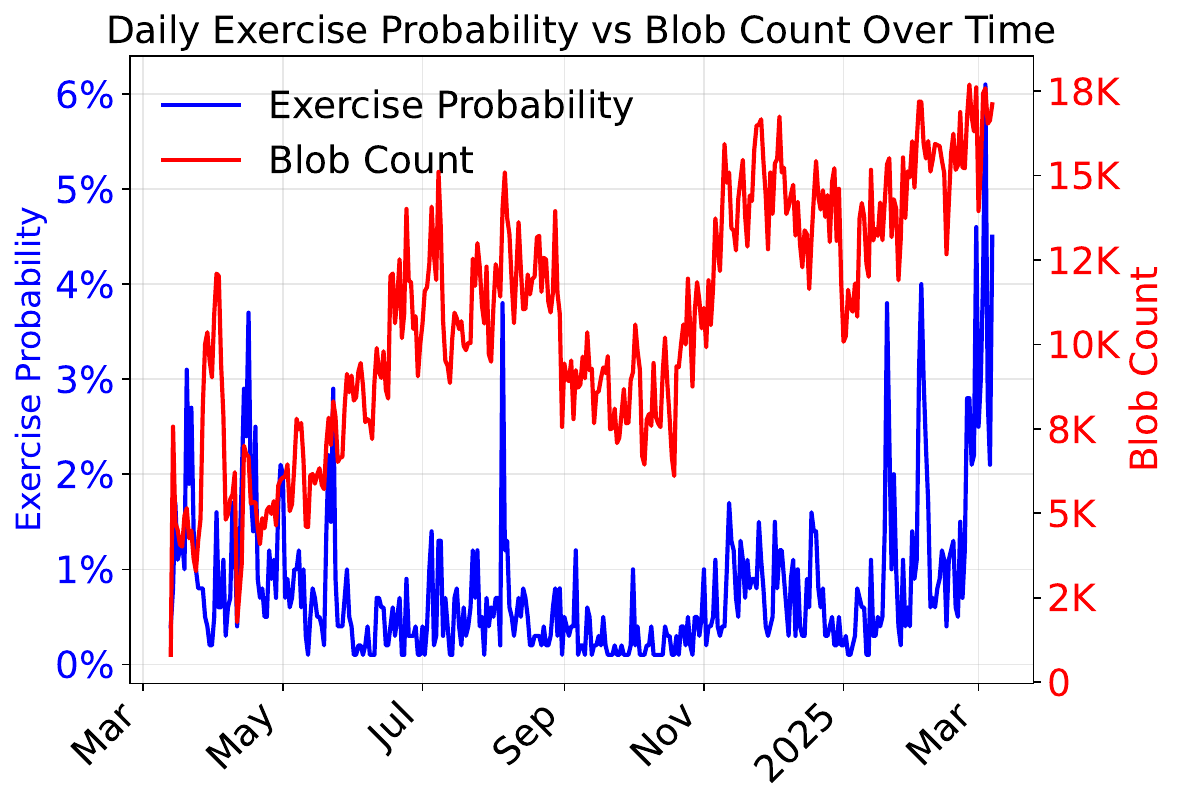}
  \caption{}
  \label{fig:blob_count}
\end{subfigure}
\caption{(a) Correlation between daily option exercise probability and average CEX-DEX flow value share. (b) Daily option exercise probability and blob count.}
\end{figure}

A correlation analysis corroborates these observations. \Cref{fig:beta} shows a strong positive relationship between the average CEX–DEX share of block value and the daily option exercise probability (Pearson’s $r = 0.846, p < 0.0001$). A beta regression with a large coefficient for CEX-DEX flow value share ($\beta = 9.85, p<0.001$) confirms that higher CEX–DEX flow dependence is strongly associated with greater option exercise. These results confirm \Cref{prop4:atomic_mev}.

\subsection{Further Considerations}
While blob fee revenue is currently minimal for builders (approximately 0.3\% of the block value) \cite{liobajasonblob,DAblobminimalfee}, in a future where blob inclusion could be monetized by builders, the potential loss of blob fee revenue at times of high blob demand could make exercising the option less frequent. 
However, as illustrated in \Cref{fig:blob_count}, data show no systematic relationship between blob demand and option exercise probability. As such, blobs currently do not provide a meaningful deterrent.

On the other hand, \emph{trailing MEV} can act as a disincentive against exercising the free option. When a slot is empty and no state transition occurs, part of MEV from pending transactions carries over to the next slot. This creates an additional incentive for building the following block, in order to capture the aggregate value from both slots. Under ePBS, a similar dynamic arises when the builder exercises the free option, rendering the slot empty and causing the transactions to return to the mempool.

Accounting for this effect when calculating option exercise probability and value reveals a sharp reduction in option use in consecutive slots.\footnote{For tractability, we assume the entire value of non-CEX-DEX order flow trails to the next slot.} For an 8-second option window, the share of blocks in which the option is profitable to exercise drop from 0.82\% to 0.3\% while the median option value falls from 0.84 ETH to 0.35 ETH. In other words, if the previous slot is empty, roughly 63\% of options that would otherwise be exercised remain unused. Trailing MEV therefore dampens the incentive to exercise the option in consecutive slots, reflecting an endogenous feature of block-building dynamics that reduces the liveness risks introduced by ePBS.

%% file: arxiv_sections/mitigations.tex
\section{Mitigations}

\label{sec:mitigations}


We analyze two explicit mitigation strategies for the free option problem under the current ePBS specification, namely shortening the option window and imposing penalties on builders who exercise the option.



\subsection{Shortening Option Window}
\begin{figure}[t]
\centering
\begin{subfigure}[t]{0.49\textwidth}
  \centering
  \includegraphics[width=\linewidth]{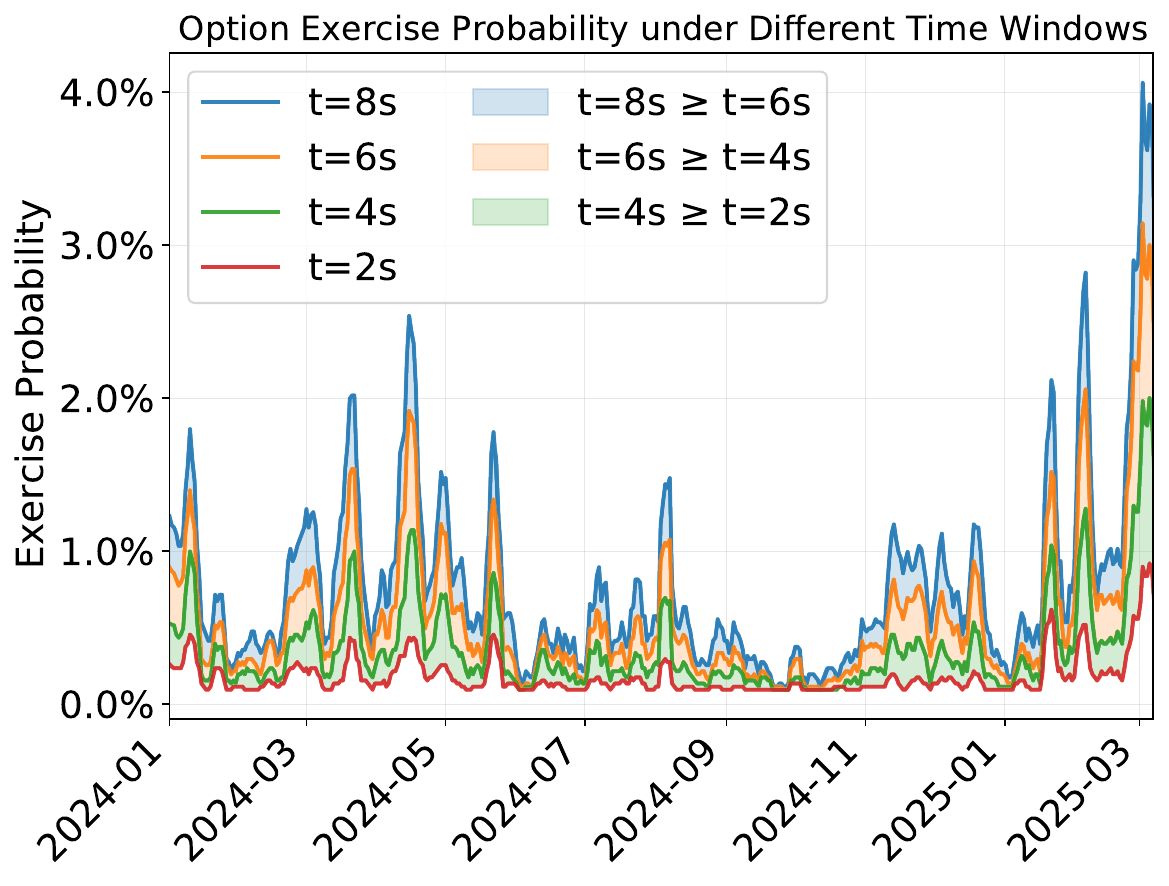} 
  \caption{}
  \label{fig:shorter_window}
\end{subfigure}\hfill
\begin{subfigure}[t]{0.49\textwidth}
  \centering
  \includegraphics[width=\linewidth]{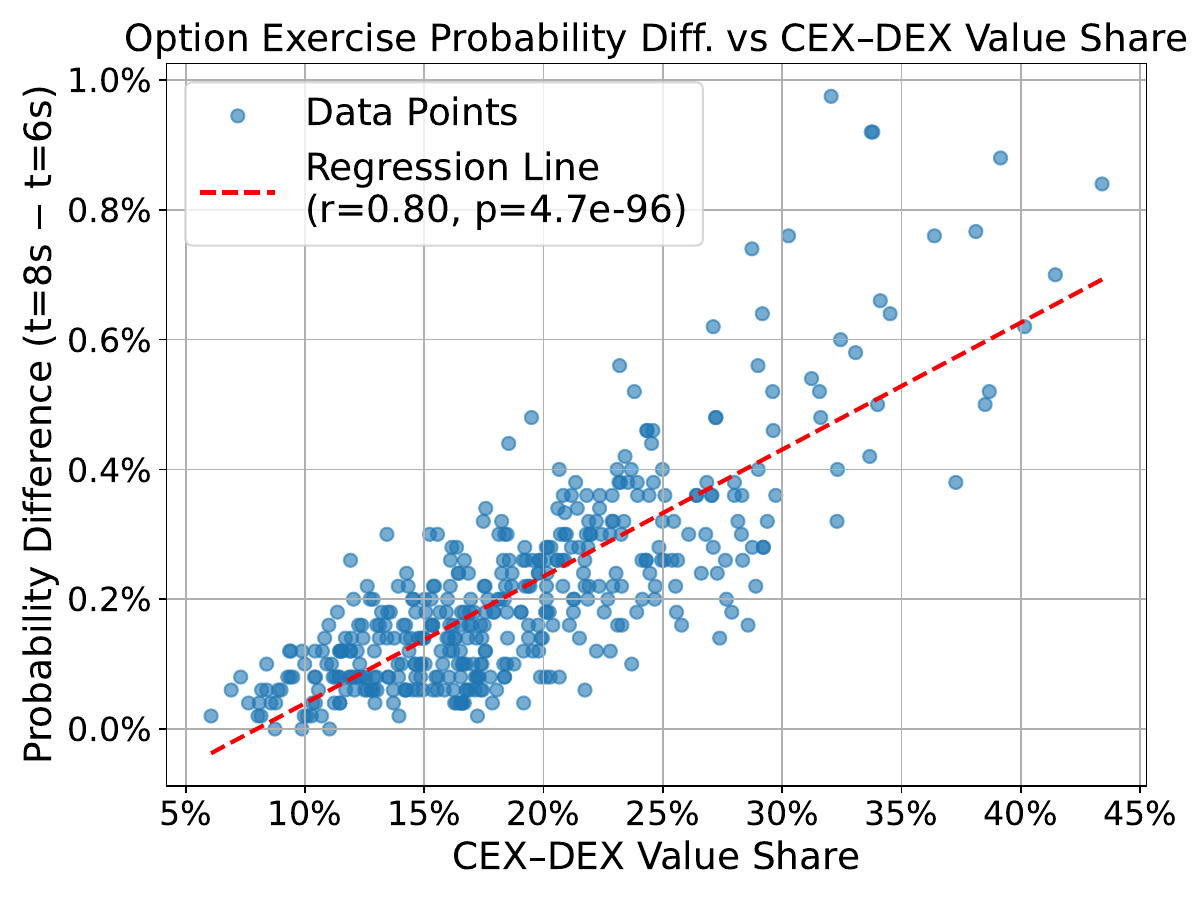}
  \caption{}
  \label{fig:shorter_window_correlation}
\end{subfigure}
\caption{(a) Option exercise probability under different free option windows (2s, 4s, 6s, 8s). (b) Correlation between the reduction in exercise probability when shortening from 8s to 6s and the CEX–DEX flow value share of block value.}
\end{figure}

As established in \Cref{prop2:return,prop3:vol}, the option value and exercise probability increase with the length of the free option window. A natural mitigation is therefore to shorten this window, i.e., moving the PTC deadlines forward closer to the slot start.

We evaluate exercise probabilities under different window lengths. Consistent with theory, shortening the window substantially reduces the exercise probability and thus reduces the liveness risk introduced by the option. In~\Cref{fig:shorter_window}, we observe that reducing the option window from 8s to 6s lowers the average exercise probability by more than 33\%; reducing it to 4s cuts the probability by over 50\%; and a 2s window reduces it by more than 77\%. Moreover, \Cref{fig:shorter_window_correlation} shows that the exercise probability difference between 8s and 6s is positively correlated with the CEX–DEX flow value share of block value. This implies that shortening the option window curtails exercise probability, most sharply in CEX–DEX flow dominated blocks where the option is most valuable.




We also observe that the effect of a shorter option window is very pronounced on the most volatile days in our sample. We detail the analyses in \Cref{appendix:high_volatility_days}. Therefore, shortening the option window is an effective mitigation, supported both theoretically and empirically. However, it comes with trade-offs. 
Tightening the PTC deadlines directly undermines one of the key benefits of ePBS: scalability. On L1, longer deadlines give validators more time to execute L1 blocks; while on L2, they allow more blobs to be propagated. Shortening the option window therefore mitigates liveness risk at the expense of execution capacity on L1 or data availability for L2s, weakening the scaling advantages that ePBS is meant to deliver.

\subsection{Penalties}


A complementary mitigation is to introduce explicit penalties for builders who exercise the option. In effect, this makes the option costly rather than free. By an argument analogous to that in~\Cref{prop4:atomic_mev}, we obtain the following result (see~\Cref{appendix:proofs} for a proof).
\begin{proposition}
 \label{prop5:penalty}
 The option value $V^*$ and the exercise probability $P^*$ decrease with the level of penalties.
 \end{proposition} 
While, penalties can be effective, as we argue next, it is important to note that they also have obvious downsides: they raise entry barriers for new builders, as they increase the capital required to buffer potential penalties. This reduction in competition can potentially, in turn, lower proposer revenue if fewer builders can bid aggressively in the ePBS auction \cite{yang2025decentralization}. Moreover, penalties raise the expected cost of blob inclusion: unless blob priority fees are sufficiently high, they do not compensate for the probability of missing PTC deadlines and being penalized. In equilibrium, either blob fees rise to cover the expected liability, or builders optimally reduce blob inclusion to avoid that risk.

\subsubsection*{Static Penalties}
We first evaluate exercise probabilities under different static penalty levels, i.e., the costs to exercise the option. As illustrated by~\Cref{fig:penalty_prob}, a cost of 0.075 ETH reduces the average exercise probability by roughly 75\%; a cost of 0.15 ETH by 83\%; and a cost of 0.5 ETH by 87\%. As with option window shortening, \Cref{fig:penalty_corr} further shows that this effect grows with CEX–DEX flow: the change in exercise probability between the no-penalty and 0.075 ETH regimes is positively correlated with the CEX–DEX share of block value. Static penalties are also effective in reducing exercise probability on high-volatility days, but are less effective in reducing option value (cf. \Cref{appendix:high_volatility_days}). 

\begin{figure}[t]
\centering
\begin{subfigure}[t]{0.49\textwidth}
  \centering
  \includegraphics[width=\linewidth]{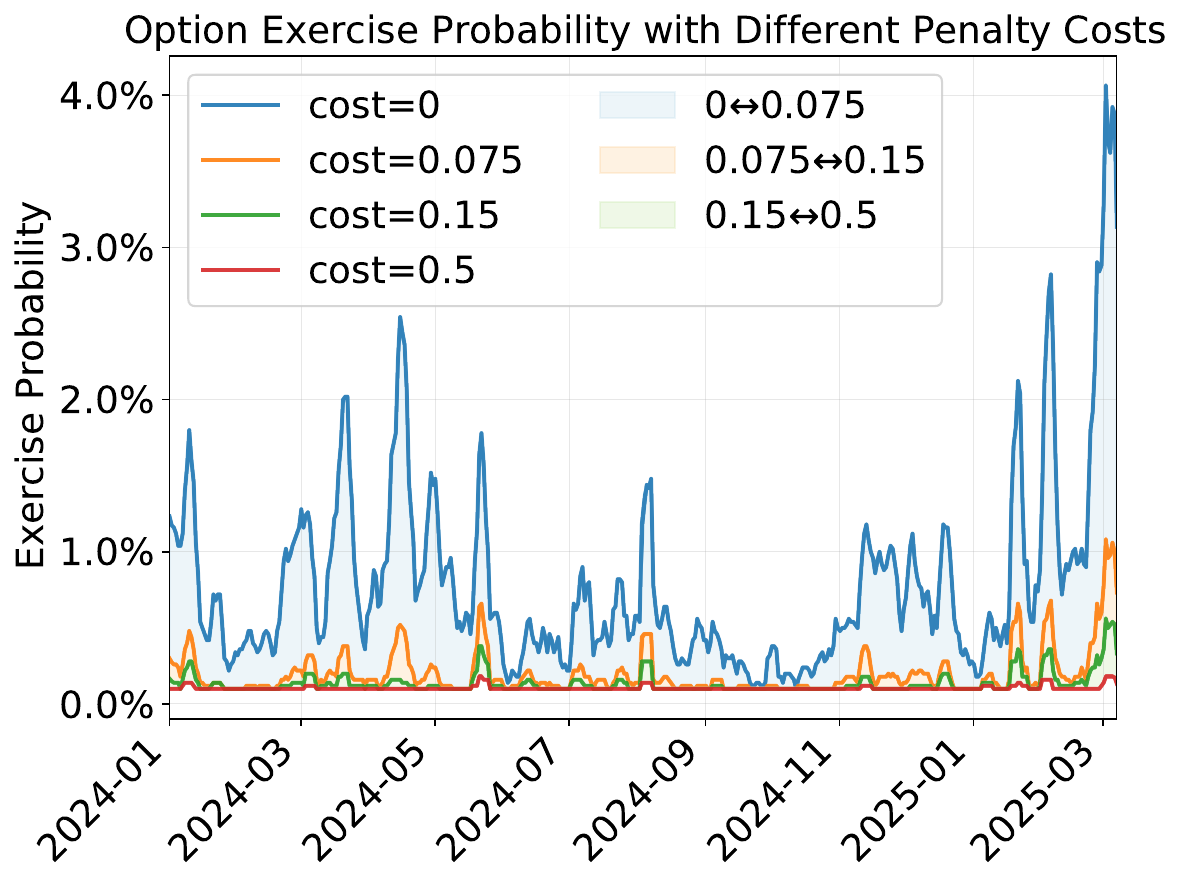} 
  \caption{}
  \label{fig:penalty_prob}
\end{subfigure}\hfill
\begin{subfigure}[t]{0.49\textwidth}
  \centering
  \includegraphics[width=\linewidth]{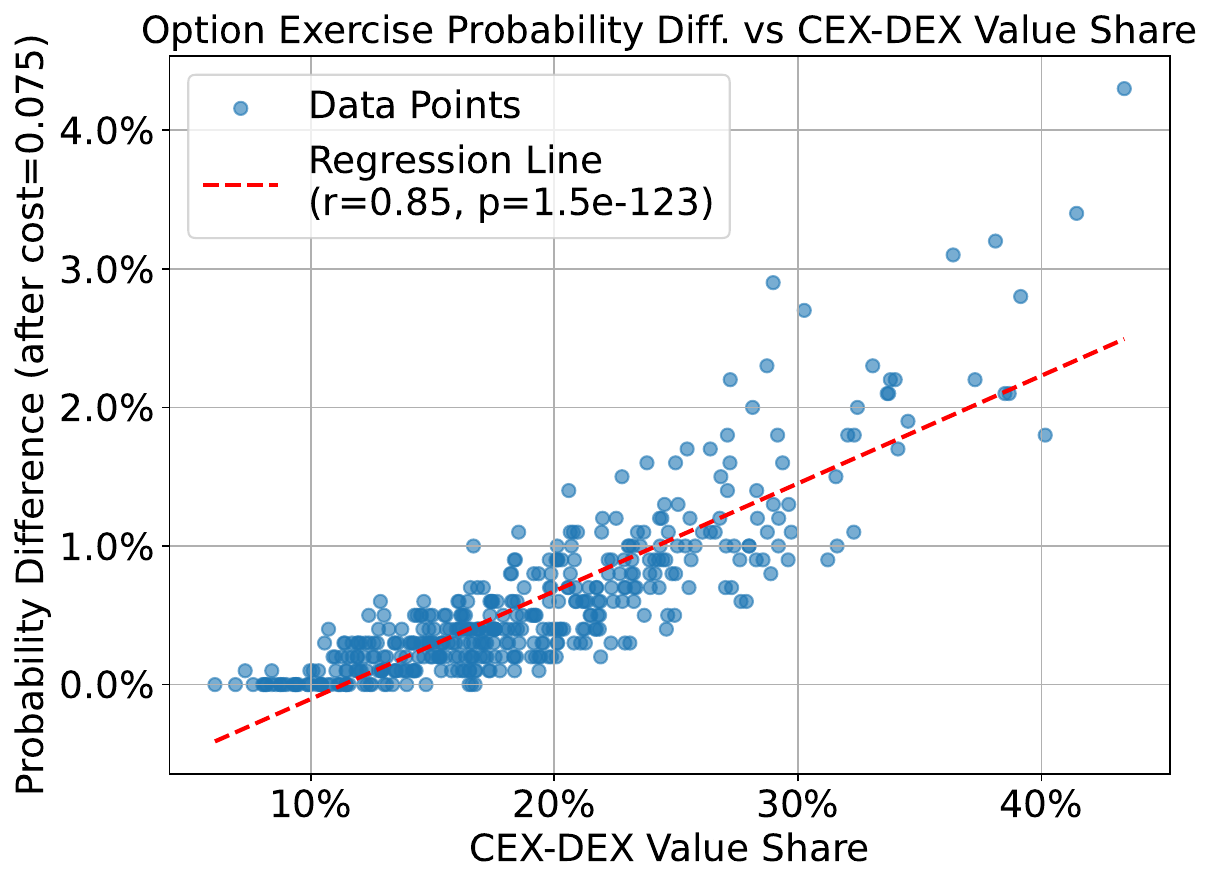}
  \caption{}
  \label{fig:penalty_corr}
\end{subfigure}
\caption{(a) Option Exercise Probabilities under different penalties (0, 0.075, 0.15, 0.5 ETH). (b) Correlation between the reduction in exercise probability when introducing a penalty of 0.075 ETH and the CEX-DEX flow value share of block value.}
\end{figure}

\subsubsection*{Dynamic Penalties}
Because market conditions vary quickly, penalty calibration should ideally adapt in real time. 
A simple idea is to index the penalty to the builder bid, on the premise that higher bids signal higher option value. In our data, however, the correlation between bids and the option exercise probability is only moderately strong (Pearson's $r = 0.516$, $p< 0.0001$). Bids are not a reliable proxy for factors such as price volatility or on-chain liquidity, which crucially determine option value. Moreover, only a handful of builders may strategically exercise the option, so their bids would not necessarily be competitive. So, the bids could fail to convey enough information for the protocol to calibrate penalties effectively. Finally, this mechanism is not resistant to off-chain agreements between the validator and the block builders. In fact, as shown in other contexts \cite{resnick2023contingent}, the revenue-maximizing mechanism is not to have penalties at all, and so there would be incentives to run an off-chain auction without using ePBS.
Another idea is to adjust penalties based on the \emph{historical exercise frequency}—the fraction of blocks missed under a given penalty. Under an i.i.d.\ assumption, with some penalty $p$, one could estimate $\pi(p)=\Pr(\text{exercise}\mid p)$ by the empirical exercise rate and choose $p$ to hit a target. In practice, markets are nonstationary, so such estimates degrade quickly. Moreover, the consensus protocol cannot observe the market conditions, such as liquidity, volatility, and order flow compositions. It can only observe the realized incidence of missed blocks. Therefore, a natural heuristic is to increase the penalty whenever the option is exercised and decrease it otherwise. This raises the question: with only such limited feedback, can we design a dynamic penalty mechanism that adapts to changing conditions and achieves near-optimal performance?

Formally, let $\alpha\in(0,1)$ denote the tolerated exercise rate, i.e., an upper bound on the fraction of rounds in which the option may be exercised, and let
\[
y_t = \mathbf{1}\{\text{option exercised at block } t \text{ under } p_t\}.
\]
We use a projected feedback update that raises the penalty after an exercise and relaxes it otherwise (cf.~\Cref{alg:pf-controller}).




\parhead{Model.} We model penalty selection as an online decision problem against a non-adaptive adversary. At each round $t=1,...,T$, an 
adversary specifies a binary response function $f_t:\mathcal X\times \mathbb R_{\geq0}\rightarrow\{0,1\}$ and a distribution $X_t$ from which $x_t\sim X_t$ is drawn. Interpret $f(x_t,p)=1$ as the option being exercised when the event is $x_t$ and $q_t(p):=E[f_t(x_t,p)]$ as the option exercise probability for the block constructed by the builder at time $t$ with penalty $p$. Before seeing the outcome $f_t$ and $x_t$, the 
protocol chooses $p_t\in\mathbb R_{\geq0}$ (based on history $\mathcal F_{t-1}$). Moreover, the 
protocol just has access to the outcome $f_t(x_t,p_t)$ but not the random variable $X_t$ nor the function $f_t$.\footnote{This is known in online convex optimization literature as one-bit feedback or bandit feedback model.}

The goal is to choose a sequence of penalties $\{p_t\}$ that minimizes the average penalty while keeping the expected exercise probability below a target $\alpha\in(0,1)$:
\begin{equation*}
    \min \frac{1}{T}\sum_{t=1}^T p_t
    \quad
    \text{s.t.}
    \quad
    E[f_t(x_t,p_t)]\leq \alpha, \text{ for }t=1,\dots,T.
\end{equation*}
We define $\{p_t^\star\}_{t=1}^T$ the best policy in hindsight as
\begin{equation*}
    p_t^\star \in\text{argmin}_{p\in\mathbb R_{\geq0}}\left\{p: E[f_t(x_t,p)]\leq \alpha\right\}. 
\end{equation*}
For an online policy with outcome penalties $\{p_t\}$, we define the \emph{dynamic cost regret} and the \emph{dynamic constraint violation regret}
\[
    R_T = E\left[\sum_{t=1}^T [p_t - p^\star_t]\right],
    \quad
    C_T = \sum_{t=1}^T \Big[E[f_t(x_t,p_t)]-\alpha\Big]_+. 
\]
In particular, 
if $R_T=o(T)$ and $C_T=o(T)$, it implies that with sufficiently big $T$, the dynamic penalty is close to the optimal policy.

A weaker constraint-violation notion is the long-run constraint-violation regret
\begin{equation*}
    LC_T = E\left[\left[\sum_{t=1}^T (f_t(x_t,p_t) -\alpha) \right]_+\right ]
\end{equation*}

We quantify how fast the per-round optimal penalties move via the \emph{path length} (also called the \emph{variation budget}) of the comparator sequence:
\[
  P_T^\star \;:=\; 1+\sum_{t=2}^T \bigl|p_t^\star - p_{t-1}^\star\bigr|.
\]
When $P_T^\star = o(T)$, the environment is nonstationary but stable on average: it may drift and even jump occasionally, but the cumulative magnitude of those movements grows sublinearly.


\parhead{Assumptions.} We impose the following conditions, standard in online convex optimization:
\begin{enumerate}[label=A\arabic*., leftmargin=*, topsep=2pt, itemsep=1pt]
  \item The adversary is non-anticipating. That is, for each $t$, $(f_t, X_t)$ may be chosen adaptively from the history $\mathcal F_{t-1}$ (including past actions and outcomes), but not from the current choice $p_t$ or the current draw $x_t$. 
  \item After round $t$ is played, the decision maker learns the outcome $f_t(x_t,p_t)$ but not $f_t$ nor $X_t$.
  \item There exists $p_\text{max}$ such that $f_t(\cdot,p_\text{max})\equiv0$.
  \item $E[f_t(x,p)]$ is non-increasing in $p$,
  and is $\mu$-strongly decreasing\footnote{If $g(p)=E[f(x,p)]$ is differentiable, is equivalent to $g'(p)\leq-\mu$.} in $[0,p_\text{max}]$, i.e., $E[(f_t(x,p)-f_t(x,p'))(p-p')]\leq -\mu (p-p')^2$ for all $p,p'\in[0,p_\text{max}]$. Also, we assume that $E[f_t]$ is $L-$Lipschitz, i.e., $|E[f_t(x,p)-f_t(x,p')]|\leq L|p-p'|$.
\end{enumerate}

\paragraph{Algorithm.}
The dynamic penalty mechanism is implemented via projected online gradient descent. Define $\phi_t(p)=\int_0^p (\alpha-q_t(z))\,dz$. Observing $f_t(x_t,p_t)$ provides a noisy gradient oracle $\partial_p\phi_t(p_t)=\alpha-q_t(p_t)$.  The explicit update rule is shown in Algorithm~\ref{alg:pf-controller}. In the following, we will utilize techniques from online convex optimization, as described in \cite{hazan2016introduction}.\footnote{Other applications of online convex optimization in DeFi can be found in \cite{angeris2024multidimensional,chitra2025curationary}.}

\begin{algorithm}[H]
\caption{Dynamic penalty, Online gradient descent (OGD)}
\label{alg:pf-controller}
\begin{algorithmic}[1]
\Require target level $\alpha\in(0,1)$, step-size $\eta_t\geq0$
\State Initialize $p_1 \gets 0$
\For{$t=1,2,\dots,T$}
  \State Play $p_t\in \mathbb R_{\geq0}$ and observe $y_t:=f_t(x_t,p_t)\in\{0,1\}$
  \State Set $g_t \gets \alpha-y_t$ and $\eta_t\leftarrow \frac{1}{\sqrt{t}}$
  \State Update $p_{t+1} \gets \Pi_{\mathbb R_{\geq0}}(p_t-\eta_t\,g_t)$
\EndFor
\end{algorithmic}
\end{algorithm}


\begin{proposition}
\label{prop:dynamic_penalty}
The Dynamic Penalty algorithm achieves dynamic cost regret and dynamic constraint violation regret of order $O\!\left(T^\frac{3}{4}\sqrt{P^\star_T+\log(T)}\right)$ and long-run constraint-violation regret $O(\sqrt{T})$ under assumptions (A1)--(A4).
\end{proposition}
The above result is concerned with the expectation of regret (see~\Cref{appendix:dynamic} for a proof). 
To control realized violations, define $m_t := f_t(x_t,p_t)-\mathbb{E}[f_t(x_t,p_t)\mid \mathcal F_{t-1},p_t]$. Then $\{m_T\}$ is a martingale with $E[m_t]=0$ and $|m_t|\le1$. Therefore, by the Azuma-Hoeffding inequality, with probability at least $1-\delta$,
\[
\left|\sum_{t=1}^T f_t(x_t,p_t)-\sum_{t=1}^T \mathbb{E}\!\left[f_t(x_t,p_t)\mid \mathcal F_{t-1},p_t\right]\right|
\;\le\; \sqrt{2T\log\!\frac{2}{\delta}}.
\]
Combining with the bound on $C_T$ gives
\[
\sum_{t=1}^T\big[f_t(x_t,p_t)-\alpha\big]_+
\;\le\; O\!\left(T^\frac{3}{4}\sqrt{P_T^\star+\log T}\right)
\;+\; \sqrt{2T\log\!\frac{2}{\delta}},
\]
with probability at least $1-\delta$.


\parhead{Empirical Performance.} As illustrated in \Cref{fig:dynamic}, with target $\alpha=0.1\%$ and a step size of ${1}/{\sqrt{7200}}$, the dynamic mechanism achieves an average exercise probability of $0.096\%$, well below the no-penalty baseline ($0.82\%$) and comparable to the level under a high static penalty of $0.5$~ETH ($0.107\%$). Importantly, it attains this outcome with an average penalty of only $0.104$~ETH. Beyond average conditions, the dynamic penalty mechanism also adapts effectively on high-volatility days. 
The option exercise probability falls from 2.9\% to 0.5\% on May 23, 2024, 
from 3.8\% to 0.6\% on August 5, 2024, 
from 3.1\% to 0.6\% on February 3, 2025, 
and from 6.1\% to 0.2\% on March 4, 2025.

\begin{figure}[t]
    \centering
    \includegraphics[width=1\linewidth]{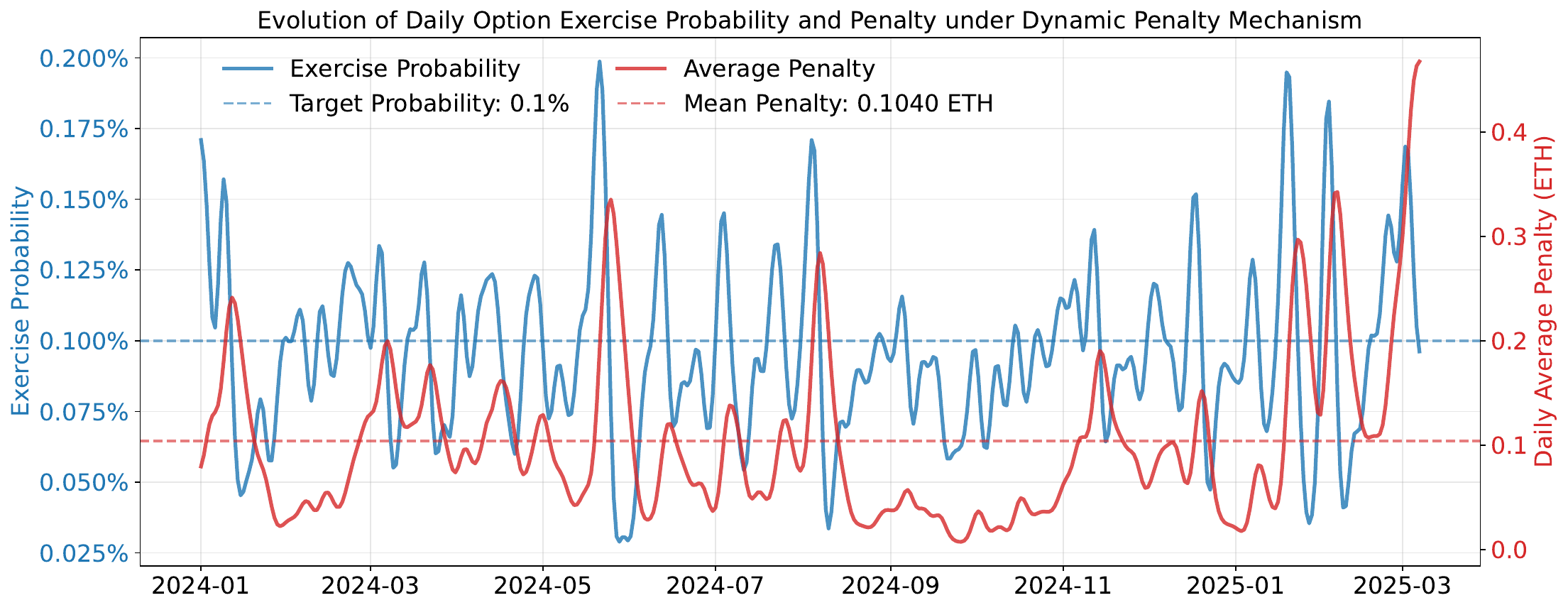}
    \caption{Daily exercise probability and average penalty under the dynamic penalty mechanism with step size $\eta_t=\tfrac{1}{\sqrt{7200}}$ and target exercise probability $\alpha=0.1\%$.}
    \label{fig:dynamic}
\end{figure}

%% file: arxiv_sections/Conclusion.tex
\section{Conclusion}

We have documented and analyzed the free option problem introduced by ePBS with dual PTC deadlines. While the design removes reliance on trusted relays, it also creates a distinct liveness risk: builders can condition payload or blob release on late-arriving external information and invalidate payloads when unfavorable. Our theoretical model shows that the option value and exercise probability increase with liquidity, volatility, and CEX-DEX MEV value share in the block. Empirically, in historical blocks, we find that although exercise probabilities are low on average, they spike in periods of high volatility—the very moments when users most rely on timely execution. Empty slots thus appear precisely when predictability is most critical, degrading user experience.

Our theoretical analysis further shows that builders exploiting the free option would deteriorate on-chain market functioning and efficiency by prolonging periods of stale prices. LPs are particularly adversely affected. 
\emph{Active} LPs forgo revenue they would have earned when prices move in their favor but DEX pools cannot adjust promptly \cite{lvr}, while \emph{passive} LPs lose trading-fee revenue. Moreover, on-chain price signals become distorted: builders exercising the option tend to ``overshoot'', inducing temporary mispricing. While such overshooting may partially compensate LPs ex-post with trading fees, it undermines the reliability of DEXes as price oracles. 
\Cref{tab:stakeholder_impact} summarizes the impact of the free option across market participants.

We also discuss mitigation strategies that require protocol-level interventions. Shortening the option window directly lowers the option exercise probability but reduces scalability. Penalties make the option costly rather than free, 
but raise entry barriers for builders. Both approaches entail trade-offs between scalability, entry barriers for builders, and protocol complexity. While reputation effects may discourage some builders from exercising the option, they cannot be relied upon in equilibrium. Protocol design should therefore provide robust incentives that hold across various market conditions.

Looking forward, it is important to recognize that the free option problem arises not only in trading assets but potentially in any setting where valuable information can arrive between the commitment and PTC deadlines—for example, in prediction markets or other real-time, time-sensitive applications. 

\input{tables/stakeholder_impact}

%% file: tables/stakeholder_impact.tex
\begin{table}[t]
\centering
\caption{Impact of the ePBS Free Option Across Market Participants}
\label{tab:stakeholder_impact}
\footnotesize
\setlength{\tabcolsep}{3pt}
\renewcommand{\arraystretch}{1.1}
\resizebox{\textwidth}{!}{%
\begin{tabular}{
  l
  @{\hspace{0.25cm}}
  >{\RaggedRight\arraybackslash}p{0.22\linewidth}
  @{\hspace{0.22cm}}
  c
  @{\hspace{0.25cm}}
  >{\RaggedRight\arraybackslash}p{0.62\linewidth}
}

\toprule
\textbf{Participant} & \textbf{Objective} & \textbf{Impact} & \textbf{Economic Mechanism} \\
\midrule
Users
& Timely execution at fair prices
& \downarrowr
& Prolonged inclusion delays lead to stale prices, increasing slippage and adverse selection risks. \\
\arrayrulecolor{gray}
\midrule
\makecell[l]{Atomic MEV\\Searchers}
& Extract atomic MEV
& \downarrowr
& Slower price discovery reduces arbitrage frequency.\\
\midrule
Active LPs
& Capture price movements and fees
& \downarrowr
& Loss of revenue when prices move in their favor but pools cannot adjust promptly. \\
\midrule
Passive LPs
& Earn trading fees
& \downarrowr
& Reduced trading activity during stale-price intervals lowers fee revenue. \\
\midrule
Builders
& Maximize block profit
& \uparrowg
& Free option enables conditional payload release, increasing expected surplus by avoiding adverse outcomes, but repeated exercise may erode reputation and proposer trust. \\

\midrule
Proposers
& Maximize block-outsourcing revenue
& \ambarrow
& Option-induced builder bids raise short-run revenue, and under trustless payments validators are paid regardless of option exercise; long-run returns may fall as market efficiency deteriorates.\\
\arrayrulecolor{black}
\bottomrule
\end{tabular}
}
\footnotesize
\textit{Legend:} \textcolor{green!60!black}{$\uparrow$}-positive impact,
\textcolor{red!70!black}{$\downarrow$}-negative impact,
\textcolor{orange!80!black}{$\updownarrow$}-ambiguous.
\end{table}

%% file: arxiv_sections/appendix.tex
\section{Proofs}
\label{appendix:proofs}
\subsection*{Proof of Proposition~\ref{prop1:liquidity}}
\begin{proof}
We assume (w.l.o.g.) that the DEX follows a constant function market maker defined by $f(X,Y)=L$. We assume that the CFMM is convex.

By the envelope theorem, we have: 
\begin{align*}\frac{\partial}{\partial L}V^*&=\frac{\partial}{\partial L}E[\max\{0,\mu+(1+r_{\tau})y^*-P_{DEX}(y^*/P_0)\}]\\&=\frac{\partial}{\partial L}\int_{\frac{\Pi_0(y^*)}{y^*}}^{\infty}(\mu+xy^*-P_{DEX}(y^*/P_0))f_{r_\tau}(x)dx\\&=-\tfrac{\partial P_{DEX}(y^*/P_0)}{\partial L}(1-P^*)
\end{align*}

where $y^*$ is the optimal position. We have 
$$P_{DEX}(y^*/P_0)=X^{old}-X^{new}=X(L,Y^{old})-X(L,Y^{old}+y^*/P_0),$$ where $X(L,Y)$ denotes the numéraire reserves of the AMM as a function of $L$ and of the reserves of the risky asset, and $X^{old}$ (resp. $Y^{old}$) denotes the reserves of the numéraire asset (of the risky asset) prior to the trade and  $X^{new}$ (resp. $Y^{old}$) the reserves after the trade.
\begin{align*}
\tfrac{\partial P_{DEX}}{\partial L}=\tfrac{\partial}{\partial L}X(L,Y^{old})-\tfrac{\partial}{\partial L}X(L,Y^{old}+y^*/P_0)>0,\end{align*}
where  the last inequality follows by convexity of the AMM.

For the probability, we have
$$\frac{\partial}{\partial L}P^*=\frac{\partial P^*}{\partial P_{DEX}}\frac{\partial P_{DEX}(y/P_0)}{\partial L}=\frac{\partial^2 V^*}{\partial (P_{DEX})^2}\frac{\partial P_{DEX}(y/P_0)}{\partial L}.$$
As observed previously, we have $$\frac{\partial P_{DEX}(y/P_0)}{\partial L}>0.$$
Moreover, $V^*$ is convex in $P_{DEX}(y^*/P_0)$, as it is the maximum over a function that is convex in $P_{DEX}$. Thus, $$\frac{\partial}{\partial L}P^*>0.$$
\end{proof}
\subsection*{Proof of Proposition~\ref{prop2:return}}
\begin{proof} Let $F^1_t$ and $F^2_t$ be two returns distributions such that the $r^1_t$ is second-order stochastic dominated by the $r^2_t$. By definition, for any convex non-decreasing and convex function $\varphi$, $E_{r_1^t\sim F^1_t}[\phi(x)]\leq E_{x\sim F^2_t}[\phi(x)]$. 
Since $r\mapsto \max\{0,\mu+r y-P_{DEX}(y/P_0)\}$ is convex, $E_{r_1^t\sim F^1_t}[\max\{0,\mu+r^1_\tau y-P_{DEX}(y/P_0)\}]\leq E_{r^2_t\sim F^2_t}[\max\{0,\mu+r^2_\tau y-P_{DEX}(y/P_0)\}]$, and so, the inequality holds for $V^*$ by taking max in both sides of the inequality. 
\end{proof}
\subsection*{Proof of Proposition~\ref{prop3:vol}}
Suppose there is no spread between the AMM price and the CEX price $P'_{DEX}(0)=P_0$.  Then the first order condition can be re-written as
$$\lambda(\tfrac{y^*}{\sigma L}-\tfrac{\mu}{\sigma y^*})=2\tfrac{y^*}{\sigma L}.$$
We need to show that the equation has a unique solution in $z^*:=\tfrac{y^*}{\sigma L}-\tfrac{\mu}{\sigma y^*}$ and this solution is monotonic in $\sigma.$ We can re-write the equation as $$\lambda(z^*)=z^*+\sqrt{(z^*)^2+4\tfrac{\mu}{\sigma}}.$$
The right hand side is strictly increasing in $z^*$ and strictly decreasing in $\sigma$. The left hand side is strictly increasing in $z^*$. So, there is at most one point of intersection which is shifting monotonically with $\sigma.$ Similar considerations also work if the DEX under-prices the asset.\qed
\subsection*{Proof of Proposition~\ref{prop4:atomic_mev}}
\begin{proof}
By the envelope theorem, we have
$\frac{\partial}{\partial \mu}V^*=\frac{\partial}{\partial \mu}E[\max\{0,\mu+r_{\tau}y^*-P_{DEX}(y^*/P_0)\}],$
where $y^*$ is the optimal position.
\begin{align*}
\frac{\partial}{\partial \mu}V^*=\frac{\partial}{\partial \mu}\int_{\frac{\Pi_0(y^*)}{y^*}}^{\infty}(\mu+xy^*-P_{DEX}(y^*/P_0))f_{r_\tau}(x)dx\\=\int_{\frac{\Pi_0(y^*)}{y^*}}^{\infty}f_{r_\tau}(x)dx=1-P^*>0.\end{align*}
Moreover, the above implies in particular that  $$P^*=1-\frac{\partial}{\partial \mu}V^*\Rightarrow\frac{\partial}{\partial \mu}P^*=-\frac{\partial^2}{\partial \mu^2}V^*,$$
which means that the exercise probability is decreasing in $\mu$, if and only if value is convex in $\mu$. As the option value is a maximum over functions that are convex in $\mu$, it is convex and the value is decreasing.
\end{proof}


 
\subsection*{Proof of Proposition~\ref{prop5:penalty}}
\begin{proof}
Consider the option value under penalty $p$:
\begin{align*}
    V^*
    &= \max_y \, E\!\left[\max\{-p,\, \mu + ry - P_{DEX}(y/P_0)\}\right] \\
    &= \max_y \Bigl\{ E\!\left[\max\{0,\, \mu + p + ry - P_{DEX}(y/P_0)\}\right] \Bigr\} - p.
\end{align*}
Analogous to proposition \ref{prop4:atomic_mev}, differentiating with respect to $p$, we obtain
\[
    \frac{\partial V^*}{\partial p} = - \Pr\bigl[\mu + p + ry < P_{DEX}(y/P_0)\bigr].
\]
Thus, the option value decreases in $p$. Moreover, note that for fixed $y$, the function 
\[
    p \mapsto \max\{0,\, \mu + p + ry - P_{DEX}(y/P_0)\}
\]
is convex in $p$. Since the pointwise maximum of convex functions is convex, $V^*$ is convex in $p$. It follows that both the option value $V^*$ and the exercise probability $P^*$ are decreasing in the level of the penalty.
\end{proof}

\section{Proof of Dynamic Penalties}
\subsection*{Proof of Proposition~\ref{prop:dynamic_penalty}}
By (A3) and the update rule, $p_t\le p_{\max}+\alpha$ for all $t$; hence we may, wlog, run OGD on the interval $[0,p_{\max}+\alpha]$ (i.e., as a projection onto this set).

\label{appendix:dynamic}
\begin{lemma} Let $d_t=\frac{1}{2}(p_t-p^\star_t)^2$ and $\Delta^\star_t := (p^\star_{t+1}-p^\star_t)$. For any steps $\{\eta_t\}_{t\geq1}$,
\begin{equation*}
    \sum_{t=1}^T \eta_t E[(\alpha-q_t(p_t))(p_t-p^\star_t)]\leq \frac{(p_{\max}+\alpha)^2}{2} + \frac{1}{2}\sum_{t=1}^T \eta_t^2 +\frac{3(p_{\max}+\alpha)}{2}(P^\star_T -1)
\end{equation*}
\end{lemma}

\begin{proof}
Projection yields
\[
|p_{t+1}-p_t^\star|^2
\le |p_t-\eta_t g_t - p_t^\star|^2
= |p_t-p_t^\star|^2 - 2\eta_t g_t(p_t-p_t^\star) + \eta_t^2 g_t^2.
\]
With $p_{t+1}-p_{t+1}^\star=(p_{t+1}-p_t^\star)-\Delta_t^\star$ and $d_t=\tfrac12(p_t-p_t^\star)^2$,
\[
d_{t+1}-d_t
\le
-\eta_t g_t(p_t-p_t^\star) + \tfrac12 \eta_t^2 g_t^2
-(p_{t+1}-p_t^\star)\Delta_t^\star + \tfrac12(\Delta_t^\star)^2.
\]
Since $p_{t+1},p_t^\star\in[0,p_{max}+\alpha]$, $|p_{t+1}-p_t^\star|\le (p_{\max}+\alpha)$ and $|\Delta_t^\star|\le (p_{\max}+\alpha)$, hence
$-(p_{t+1}-p_t^\star)\Delta_t^\star \le (p_{\max}+\alpha)|\Delta_t^\star|$ and $(\Delta_t^\star)^2\le (p_{\max}+\alpha)|\Delta_t^\star|$.
Also $g_t^2\le 1$. Therefore
\[
d_{t+1}-d_t \le -\eta_t g_t(p_t-p_t^\star) + \tfrac12 \eta_t^2 + \tfrac{3(p_{\max}+\alpha)}{2}|\Delta_t^\star|.
\]
Taking expectations, using $E[g_t\mid\mathcal F_{t-1},p_t]=\alpha-q_t(p_t)$, and summing $t=1$ to $T$ gives
\[
\sum_{t=1}^T \eta_t\,E\!\big[(\alpha-q_t(p_t))(p_t-p_t^\star)\big]
\le
E[d_1] + \tfrac12\sum_{t=1}^T \eta_t^2 + \tfrac{3(p_{\max}+\alpha)}{2}\sum_{t=1}^T E|\Delta_t^\star|.
\]
Now $d_1\le (p_{\max}+\alpha)^2/2$ and $\sum_{t=1}^T |\Delta_t^\star|=P_T^\star-1$. 
\end{proof}

By monotonicity of $q_t$ and Lipschitzness, i.e., $q_t$ is nonincreasing and $\mu$-strongly decreasing, we have for all $t$
\[
\big(q_t(p_t)-\alpha\big)_+^2 \leq L (\alpha-q_t(p_t))(p_t-p^\star_t),
\]
and so,
\begin{align*}
    \sum_{t=1}^T \eta_t E[(q_t(p_t)-\alpha)^2] \leq L\left (\frac{(p_{\max}+\alpha)}{2} + \frac{1}{2}\sum_{t=1}^T \eta_t^2 +\frac{3(p_{\max}+\alpha)}{2}(P^\star_T -1) \right).
\end{align*}
By Cauchy-Schwarz,
\begin{align*}
E[C_T]
&= \sum_{t=1}^T \bigl(q_t(p_t)-\alpha\bigr)_+ \\
&\le \left(\sum_{t=1}^T \frac{1}{\eta_t}\right)^\frac{1}{2}
   \left(\sum_{t=1}^T \eta_t\,E\bigl[(q_t(p_t)-\alpha)^2\bigr]\right)^\frac{1}{2} \\
&= O\!\left(T^\frac{3}{4}\sqrt{P_T^\star+\log(T)}\right),
\end{align*}
where the hidden constant depends on $L$ and $p_{\max}$.

Finally, by strong monotonicity, $|p_t-p_t^\star|\leq\frac{1}{\mu} |q_t(p_t)-\alpha|$, so
\[
E[R_T] =  E\left[\sum_{t=1}^T [p_t - p^\star_t]\right] \leq \frac{1}{\mu}  E\left[\sum_{t=1}^T|q_t(p_t)-\alpha|\right] =   O\left (T^\frac{3}{4}\sqrt{P^\star_T+\log(T)}\right).
\]

Now, let's bound the long-run constraint-violation regret. To do so, first observe that $f_t(x_t,p_t) - \alpha \leq \frac{1}{\eta_t} (p_{t+1}-p_t)$. Therefore,
\begin{align*}
    \left[\sum_{t=1}^T (f_t(x_t,p_t) - \alpha)\right]_+ &\leq \left [\sum_{t=1}^T \frac{1}{\eta_t} (p_{t+1}-p_t)\right ]_+ \\
    &=\left [\frac{p_{T+1}}{\eta_{T}} -\frac{p_1}{\eta_1} - \sum_{t=1}^{T-1}\left(\frac{1}{\eta_{t+1}}-\frac{1}{\eta_t}\right)p_{t+1}\right]_+\\
    & \leq p_{T+1}\sqrt{T} = O(\sqrt{T})\text{ since }
p_t\leq p_{\max}+\alpha.
\end{align*}

\section{High-volatility Days}
\label{appendix:high_volatility_days}
\begin{figure}[t]
\centering
\includegraphics[width=1\linewidth]{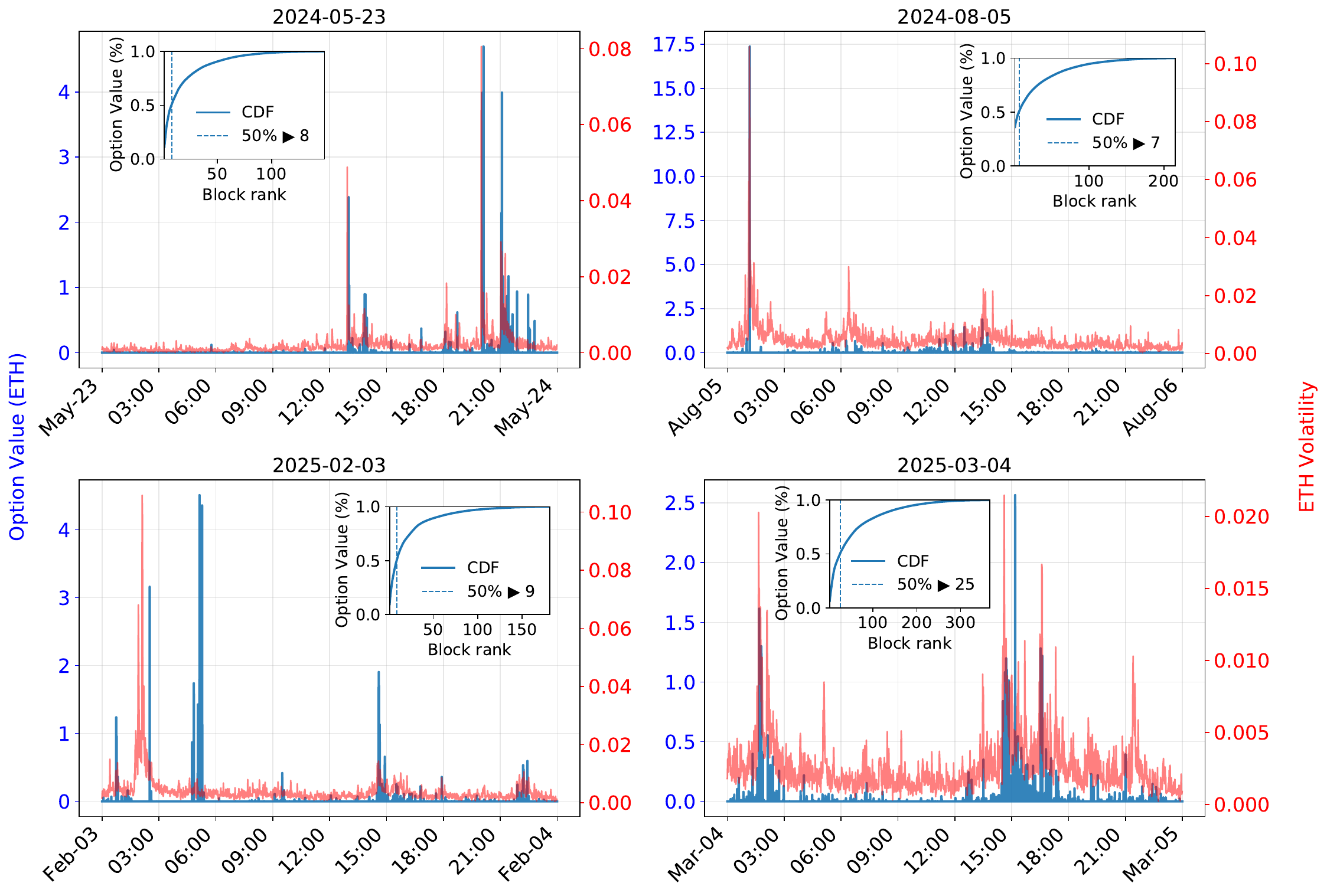}
\caption{Intraday option value and 1-minute ETH volatility on high-volatility days. Insets show the cumulative distribution function (CDF) of daily option value across blocks, illustrating how few blocks account for total option value.}
\label{fig:high_vola_days}
\end{figure}

Four high-volatility days stand out with extreme option exercise probabilities and option values: May 23, 2024; August 5, 2024; February 3, 2025; and March 4,  2025. On each of these dates, the exercise probability exceeded 3.5\% and the total daily option value surpassed 40 ETH.

\Cref{fig:high_vola_days} plots the intraday time series of realized option values alongside ETH price volatility. We see that most bursts of option value are mirrored by spikes in ETH volatility. The figure also highlights the extreme concentration of option value on just a few blocks in these days. For example, on August 5, 2024, only seven blocks accounted for more than 50\% of that day’s total option value. These patterns further reinforce our theory.

An exception occurs around 5 AM on February 3, 2025, where a burst of option value was not mirrored by ETH volatility. Five blocks with extreme option values (block \texttt{21763801, 21763802, 21763803, 21763833, 21763835}) each contain a very large CEX-DEX trade for ETH/USDC with sizes over 1M USD.\footnote{Transaction hash of one of these trades: \href{https://etherscan.io/tx/0xf7fea6501fa6b60769817be1ede76ca1883b098eddf622a8e2106caf6a8fe978}{\texttt{0xf7fea6501fa6b60769817be1ede76ca1883b098eddf622a8e210\\6caf6a8fe978}.}} These trades' positions were profitable around the block time, and the searcher payments were close to the contemporaneous markouts. However, the position deteriorated sharply 8 seconds later, making the options highly profitable. This pattern is best explained by the substantial price impact of the searchers’ own hedge trades on CEX: trade sizes were so large that, by the time hedging was completed, prices had moved against their positions. Thus, the extreme option values at this period were driven by self-imposed CEX price impact rather than general market volatility.

The effects of mitigation strategies discussed in \Cref{sec:mitigations} are pronounced on these volatile days. \Cref{tab:mitigation_shorter_window_volatile_days} presents daily aggregate option values and average option exercise probability on these volatile days under option windows of 2, 4, 6, and 8 seconds. We observe that both the option value and the exercise probability shrink under shorter option windows.

\Cref{tab:mitigation_penalties_volatile_days} shows outcomes under penalties of 0 ETH (no-penalty), 0.0075 ETH, 0.15 ETH, and 0.5 ETH on these days. While penalties substantially reduce the frequency of exercise, they are far less effective at lowering daily aggregate option values. This is because a handful of blocks contain extremely valuable options that dwarf the penalties, rendering them insufficient as a deterrent.

\sisetup{
  table-number-alignment = center,
  table-column-width = 1.1cm, 
  round-mode = places
}

\begin{table}[t]
\caption{Daily aggregate option value and average exercise probability under different option windows (top table) and penalty levels (bottom table) on high-volatility days.}
\label{tab:mitigation_shorter_window_volatile_days}
\centering
\resizebox{\textwidth}{!}{%
\begin{tabular}{
  l c
  S[round-precision=2] S[round-precision=2] S[round-precision=2] S[round-precision=2]
  S[round-precision=1, table-format=2.1]S[round-precision=1, table-format=2.1]S[round-precision=1, table-format=2.1]S[round-precision=1, table-format=2.1]
}
\toprule
 & & \multicolumn{4}{c}{Aggregate Option Value (ETH)} 
   & \multicolumn{4}{c}{Average Exercise Probability (\%)} \\
\cmidrule(lr){3-6} \cmidrule(lr){7-10}
 & {$\tau$} 
  & {8s} & {6s} & {4s} & {2s} 
  & {8s} & {6s} & {4s} & {2s} \\
{Date} \\  
\midrule
2024-05-23 &  & 41.25 & 28.05 & 19.99 & 6.67 & 2.9 & 2.4 & 1.7 & 0.9 \\
2024-08-05 &  & 48.16 & 20.27 & 11.60 & 2.81 & 3.8 & 3.0 & 2.1 & 0.8 \\
2025-02-03 &  & 46.84 & 25.76 & 15.83 & 5.14 & 3.1 & 2.3 & 1.5 & 0.8 \\
2025-03-04 &  & 40.04 & 26.43 & 15.58 & 5.61 & 6.1 & 4.5 & 3.0 & 1.6 \\
\bottomrule
\end{tabular}
}
\end{table}

\sisetup{
  table-number-alignment = center,
  table-column-width = 1.05cm,
  round-mode = places
}

\begin{table}[t]
\label{tab:mitigation_penalties_volatile_days}
\centering
\resizebox{\textwidth}{!}{%
\begin{tabular}{
  l c
  S[round-precision=2] S[round-precision=2] S[round-precision=2] S[round-precision=2]
  S[round-precision=1, table-format=2.1]S[round-precision=1, table-format=2.1]S[round-precision=1, table-format=2.1]S[round-precision=1, table-format=2.1]
}
\toprule
 & & \multicolumn{4}{c}{Aggregate Option Value (ETH)}
   & \multicolumn{4}{c}{Average Exercise Probability (\%)} \\
\cmidrule(lr){3-6} \cmidrule(lr){7-10}
& \multicolumn{1}{l}{Penalty}
  & \multicolumn{1}{c}{0} & \multicolumn{1}{c}{0.075} & \multicolumn{1}{c}{0.15} & \multicolumn{1}{c}{0.5}
  & \multicolumn{1}{c}{0} & \multicolumn{1}{c}{0.075} & \multicolumn{1}{c}{0.15} & \multicolumn{1}{c}{0.5} \\
{Date} & (ETH) \\
\midrule
2024-05-23 &  & 41.25 & 34.18 & 29.96 & 20.21 & 2.9 & 1.4 & 0.9 & 0.4 \\
2024-08-05 &  & 48.16 & 38.13 & 32.72 & 22.30 & 3.8 & 1.6 & 1.0 & 0.3 \\
2025-02-03 &  & 46.84 & 38.68 & 34.47 & 23.25 & 3.1 & 1.3 & 0.8 & 0.4 \\
2025-03-04 &  & 40.04 & 25.87 & 19.97 &  8.84 & 6.1 & 1.8 & 1.1 & 0.3 \\
\bottomrule
\end{tabular}
}
\end{table}